\documentclass[jcs,crcready]{iosart1c}

\usepackage[T1]{fontenc}
\usepackage{times}%
\usepackage[numbers]{natbib}
\usepackage{amsmath,amssymb,latexsym,amsthm}
\usepackage{enumerate,paralist}
\usepackage{booktabs}
\usepackage{xspace}
\usepackage{url}
\usepackage{multirow}
\usepackage{version}
\usepackage{caption}
\usepackage{subcaption}
\usepackage{enumitem} 
\usepackage{tabularx,tabulary}
\usepackage{tikz}
\usetikzlibrary{positioning,arrows,automata,shapes,shadows,fit}

\usepackage[normalem]{ulem}

\newcommand{\eofh}{E_{\min}}
\newcommand{\eofhstar}{E_{\max}}
\newcommand{\set}[1]{\left\{#1\right\}}
\newcommand{\card}[1]{\left|#1\right|}

\newcommand{\setup}{\ensuremath{\mathsf{SetUp}}\xspace}
\newcommand{\derive}{\ensuremath{\mathsf{Derive}}\xspace}
\newcommand{\pub}{\mathsf{Pub}}
\newcommand{\prf}{\mathcal{F}}

\newcommand{\T}{\mathcal{T}}

\newcommand{\C}{\mathcal{C}}
\newcommand{\poset}{\mathcal{P}}
\newcommand{\dset}[2]{\mathord{\downarrow}_{#2}(#1)}
\newcommand{\uset}[2]{\mathord{\uparrow}_{#2}(#1)}
\newcommand{\bridge}[1]{\alpha(#1)}

\newcommand{\generickaf}{\psi}
\newcommand{\bestkaf}[1]{\phi_{#1}}
\newcommand{\gamfunc}[1]{\gamma_{#1}}

\newcommand{\gamfuncp}{\gamma_{\poset}}

\newcommand{\weightfuncp}{\omega_{\poset}}

\newcommand{\parentt}[1]{\mathrm{par}_{\mathcal{T}}(#1)}
\newcommand{\parent}[2]{\mathrm{par}_{#2}(#1)}

\newcommand{\numsecrets}[2]{{\cal S}(#2, #1)}

\newcommand{\getsr}{\gets_{\$}}

\newcommand{\keysp}{\mathcal{K}}
\newcommand{\cD}{\mathcal{D}}
\newcommand{\cA}{\mathcal{A}}
\newcommand{\Expt}{\mathrm{Expt}}
\newcommand{\Adv}{\mathrm{Adv}}
\newcommand{\kist}{{\sf kist}}
\newcommand{\ExpCorrupt}{\mathit{Corrupt}}
\newcommand{\ExpKeys}{\mathit{Keys}}

\newcommand{\opc}[1]{#1_{\C}}
\newcommand{\leqslantc}{\opc{\leqslant}}

\newcommand{\lessdotc}{\opc{\lessdot}}

\newcommand{\opt}[1]{#1_{\mathcal{T}}}
\newcommand{\leqslantt}{\opt{\leqslant}}
\newcommand{\geqslantt}{\opt{\geqslant}}

\newcommand{\gtrdott}{\opt{\gtrdot}}
\newcommand{\lesst}{\opt{<}}
\newcommand{\gtrt}{\opt{>}}
\newcommand{\nlesst}{\opt{\nless}}
\newcommand{\ngtrt}{\opt{\ngtr}}

\newcommand{\intn}{\mathcal{I}}

\newtheorem{theorem}{Theorem}
\newtheorem{lemma}{Lemma}
\newtheorem{proposition}{Proposition}

\newtheorem{corollary}{Corollary}

\newtheorem{definition}{Definition}
\newtheorem{remark}{Remark}

\pubyear{0000}
\volume{0}
\firstpage{1}
\lastpage{1}

\begin{document}


\begin{frontmatter}    

\title{Cryptographic Enforcement of Information Flow Policies without Public Information via Tree Partitions%
       \thanks{This paper generalizes and extends our earlier results~\cite{CrFaGuJoPo14,CrFaGuJo15a}. 
	       In particular, we define a new form of enforcement scheme that subsumes chain-based~\cite{CrFaGuJo15a} and tree-based enforcement schemes~\cite{CrFaGuJoPo14}. We generalize results specific to these earlier schemes in order to support our more general framework.} 
	      }
\runningtitle{Enforcement of Information Flow Policies without Public Information}

\maketitle

\author[A]{\fnms{Jason} \snm{Crampton}\thanks{Corresponding author: Information Security Group, Royal Holloway, University of London, Egham, TW20 9QY, Egham; +44 1784 443117; \url{jason.crampton@rhul.ac.uk}}},
\author[A]{\fnms{Naomi} \snm{Farley}},
\author[A]{\fnms{Gregory} \snm{Gutin}},
\author[A]{\fnms{Mark} \snm{Jones}}, and
\author[B]{\fnms{Bertram} \snm{Poettering}}
\runningauthor{Crampton, Farley, Gutin, Jones, Poettering}
\address[A]{Royal Holloway, University of London}
\address[B]{Ruhr University Bochum}

\begin{abstract}
We may enforce an information flow policy by encrypting a protected resource and ensuring that only users authorized by the policy are able to decrypt the resource.
In most schemes in the literature that use symmetric cryptographic primitives, each user is assigned a single secret and derives decryption keys using this secret and publicly available information.
Recent work has challenged this approach by developing schemes, based on a chain partition of the information flow policy, that do not require public information for key derivation, the trade-off being that a user may need to be assigned more than one secret.
In general, many different chain partitions exist for the same policy and, until now, it was not known how to compute an appropriate one.

In this paper, we introduce the notion of a tree partition, of which chain partitions are a special case.
We show how a tree partition may be used to define a cryptographic enforcement scheme and prove that such schemes can be instantiated in such a way as to preserve the strongest security properties known for cryptographic enforcement schemes.
We establish a number of results linking the amount of secret material that needs to be distributed to users with a weighted acyclic graph derived from the tree partition.
These results enable us to develop efficient algorithms for deriving tree and chain partitions that minimize the amount of secret material that needs to be distributed.
\end{abstract}


\end{frontmatter}


\section{Introduction}\label{sec:intro}

Access control is a fundamental security service in modern computing systems and seeks to restrict the interactions between users of the system and the resources provided by the system.
Traditionally, access control is policy-based, in the sense that a policy is defined by the resource owner(s) specifying those interactions that are authorized.
An attempt by a user to interact with a protected resource, typically called an \emph{access request}, is evaluated by a trusted software component, the \emph{policy decision point} (or \emph{authorization decision function}), to determine whether the request should be permitted (if authorized) or denied (otherwise).
The use of a policy decision point is entirely appropriate when we can assume the policy will be enforced by the same organization that defined it.
However, use of third-party storage, privacy policies controlling access to personal data, and digital rights management all give rise to scenarios where this assumption does not hold.

\emph{Cryptographic access control} provides an alternative way of regulating access to data objects and has attracted considerable attention in recent years.
In this setting, data objects are encrypted and appropriate decryption keys are issued to authorized users.
Research into cryptographic access control began with the seminal work of Akl and Taylor~\cite{AkTa83}, and has seen a resurgence of interest in recent years.
For instance, there has been a considerable amount of research into attribute-based encryption~\cite{BeSaWa07,GoPaSaWa06}, which is regularly used to support access control (see~\cite{YuWaReLo10}, for example).
Attribute-based encryption is based on asymmetric cryptographic primitives, which means that any user is able to control read access to data (by encrypting), while only authorized users may decrypt.
However, access control policies can also be enforced using symmetric cryptographic primitives (often a cheaper alternative to their asymmetric counterparts).
Typically, in this scenario, a specific user -- the data owner -- encrypts all data objects before transmitting them to a storage provider that is only trusted to store data correctly.
Users are able to retrieve data objects from the storage provider (in encrypted form) and only authorized users should be able to decrypt them.

In the symmetric setting, the focus of research has been on enforcing information flow policies~\cite{BeLa76}, not least because many access control requirements may be articulated as information flow policies.
An information flow policy is defined by a partially ordered set of security labels and a function mapping each user and data object to a security label.
A user is authorized to read any data object associated with a security label that is less than or equal to that of the user.

Generally, it is undesirable to explicitly provide a user with all the keys she requires to decrypt protected objects.
Instead, a user is given a small number of secrets from which she is able to derive all keys required.%
\footnote{We could, of course, simply view a set of secrets as a single secret and consider the amount of storage required by that secret.  
	  However, it is more convenient for the analysis later in the paper to consider a set of secrets and the number of elements in that set.}
Hence, a common feature of cryptographic enforcement schemes for information flow policies is the derivation of decryption keys (since possession of the decryption key for label $\ell$ implies authorization for the decryption key for any label $\ell'$ less than $\ell$).
Informally, each security label is associated with a secret (which is issued to every user assigned to that security label) from which decryption keys for all subordinate security labels may be derived.
The scheme may also publish additional information in order to support key derivation.

Therefore, the challenge is to compute efficiently the secrets and decryption keys associated with each security label, subject to constraints on the size of relevant parameters.
Thus a cryptographic enforcement scheme may be characterized by%
\begin{inparaenum}[(i)]
 \item the number of secrets each user is given,
 \item the total number of secrets issued to users,
 \item the amount of auxiliary (public) information required for key derivation, and
 \item the computational effort required for key derivation.
\end{inparaenum}

Many schemes in the literature are space-efficient (on the user side) by providing each user with a single secret (see, for example,~\cite{AtBlFaFr09}), the trade-off being that the amount of public information and derivation time may be substantial.
Moreover, the public information must either be transmitted to each user or made available on some publicly accessible server, both possibilities giving rise to concerns either about costs of transmission and local storage, or availability and authenticity of the information.

Crampton, Daud and Martin~\cite{CrDaMa10} introduced the concept of a \emph{chain-based} cryptographic enforcement scheme, which requires no public information but may require users to store more than one secret.
Subsequent work has established that secure instantiations of chain-based schemes exist~\cite{FrPa11,FrPaPo13}.
Chain-based schemes are based on a decomposition of the poset of security labels into disjoint chains (that are, in some appropriate sense, compatible with the poset).
Informally, the secrets associated with the labels in each chain may be derived in a top-down manner and each user is issued with a number of secrets, at most one from each chain.
Thus the number of secrets required by a user is no greater than the number of chains in the decomposition, which is significantly better, generally, than the naive solution of supplying each user with every secret for which she is authorized.

The motivation for the work in this paper can be summarized in two observations.  
First, there are, in general, many different ways to instantiate a chain-based scheme for a given information flow policy, each instantiation being defined by a particular chain partition of the partially ordered set used to specify the policy.  
The number of secrets and the amount of computation required to derive decryption keys in a given instantiation crucially depends on the chain partition chosen. 
However, existing work in the literature assumes that the chain partition is given as part of the input to the algorithm that outputs the secrets and decryption keys.  
One of the questions we address (in Section~\ref{sec:chain-based-schemes}), therefore, is how to compute the ``best'' chain partition (with respect to some suitable metric) with which to instantiate a chain-based scheme.  
Our second observation is that each security label has at most one parent in the chain decomposition.  
The question we address (in Section~\ref{sec:schemes}) is whether it is possible to generalize chain-based schemes to tree-based schemes, given that each element in a tree also has at most one parent.

Our first set of contributions is associated with the novel concept of a \emph{tree partition} of an information flow policy, from which we define the notion of a \emph{forest-based} cryptographic enforcement scheme for information flow policies.
We prove results establishing how the total number of secrets to be issued to users varies with the structure of the forest and demonstrate that an instantiation of our scheme retains the security property of strong key indistinguishability introduced by Freire, Paterson and Poettering~\cite{FrPaPo13}.
We design and analyze an efficient algorithm for computing a forest that minimizes the total number of issued secrets. 
This work generalizes our previous work on tree-based enforcement schemes~\cite{CrFaGuJoPo14}.
In addition, the more general framework enables us to simplify the techniques and formal exposition.

Our second set of contributions is based on specializing our generic scheme to chain-based schemes.%
\footnote{One disadvantage with forest-based schemes is that one cannot, in general, simultaneously minimize the number of secrets issued on a per-user basis and the total number of secrets issued to users.
Thus, chain-based schemes are still relevant, even though, in general, a forest-based scheme for the same policy will require fewer secrets in total to be issued.}
We prove that the total number of secrets issued is determined by the number of bottom elements of the chains in the chain partition (Lemma~\ref{lem:number-of-keys-from-bottom-elements}).
This, in turn, allows us to prove (Theorem~\ref{thm:chain-partition-only-requires-w-chains}) there exists a chain partition that simultaneously minimizes the number of secrets that need to be issued and the number of chains in the partition (and thus the number of keys each user is required to store).
The last result is of practical importance, since the number of chains provides a tight upper bound on the number  of secrets required by any user. Moreover, the result is somewhat unexpected, as it is not usually possible to simultaneously minimize two different parameters.
Our main contribution (Theorem~\ref{thm:main-theorem} and Section~\ref{sec:chain-partition-requiring-widehat-k-keys}) is to develop an efficient algorithm that enables us to find a chain partition such that the total number of distributed secrets and the number of chains are minimized (with respect to all chain partitions).
Our algorithm is based on finding a minimum cost flow in a network whose construction is based on the technical results in Sections~\ref{sec:schemes}--\ref{sec:chain-based-schemes}.

Overall, then, the contributions of this paper generalize and unify existing work on tree- and chain-based schemes using the novel concept of a tree partition and a forest-based enforcement scheme.  
Central to our work are the results in Section~\ref{sec:tree-based-schemes}, which enable us to link two different characterizations of the additional secrets required, thereby allowing us to describe existing schemes using trees and chains within a single framework and to generalize tree-based enforcement schemes to forest-based schemes.
An important consequence of our results is that there now exist efficient methods for instantiating cryptographic enforcement schemes that require no public information.
We thereby provide rigorous foundations for the development of efficient chain-based enforcement schemes.

The remainder of the paper is organized as follows.
In Section~\ref{sec:background}, we provide the relevant background on cryptographic enforcement schemes, and formally identify the problem.
We also discuss related work, including preliminary versions of the ideas presented in this paper~\cite{CrFaGuJo15a,CrFaGuJoPo14}. 
Then, in Section~\ref{sec:schemes}, we formally define a tree partition and a forest-based cryptographic enforcement scheme for an information flow policy.
We establish some important results connecting the structure of a given forest and the total number of secrets required by the associated cryptographic enforcement scheme.
We also establish that there exist secure instantiations of our scheme and briefly discuss cryptographic primitives that would be suitable for such an instantiation.
In Section~\ref{sec:tree-based-schemes}, we use the theoretical results of Section~\ref{sec:schemes} to develop an efficient algorithm for computing the best tree partition, in terms of the total amount of secret material required.
In Section~\ref{sec:chain-based-schemes}, we prove that there exists a chain-based enforcement scheme in which no user requires more than $w$ keys, where $w$ is the width of the information flow policy; and that the total number of issued secrets in a chain-based enforcement scheme is determined entirely by the number of bottom elements of the chain partition. 
These results, however, are not constructive \emph{per se}.
Accordingly, we also develop an efficient algorithm to derive the best chain partition.
We conclude the paper in Section~\ref{sec:conclusion} with a summary of our contributions and some ideas for future work.

\section{Information Flow Policies}\label{sec:background}
 
We first recall some basic definitions from discrete mathematics and establish some notation.
We then define what is meant by an information flow policy~\cite{BeLa76} and discuss how such policies may be enforced using cryptographic mechanisms.

A \emph{partially ordered set} (or \emph{poset}) is a pair $\poset = (X,\leqslant)$, where $\leqslant$ is a reflexive, anti-symmetric, transitive binary relation on a finite set $X$.
\begin{itemize}
  \item We write $x < y$ to indicate $x \leqslant y$ and $x \ne y$, and we may write $x \geqslant y$ whenever $y \leqslant x$. 
  \item We say $x$ \emph{covers} $y$, or $x$ is a \emph{parent} of $y$, denoted $y \lessdot x$, if $y < x$ and there does not exist $z \in X$ such that $y < z < x$. An element $x\in X$ is  \emph{maximal} if it has no parents. 
  \item The \emph{Hasse diagram} of $\poset$ is the directed acyclic graph $H(\poset) = (X,\eofh)$, where the (directed) \emph{edge} $xy \in \eofh$ if and only if $y \lessdot x$.
	We will also make use of the directed acyclic graph $H^*(\poset) = (X,\eofhstar)$, where $xy \in \eofhstar$ if and only if $y < x$.	
	Representing the covering relation as an acyclic digraph the Hasse diagram provides a minimal amount of information required to reconstruct the full order relation.%
	 \footnote{The Hasse diagram $H(\mathcal{P})=(X, \lessdot) = (X,\eofh)$ is a unique representation of the poset ${\mathcal{P} = (X,\leqslant)}$. 
		   Conversely, as the Hasse diagram $H(\poset)$ of a poset $\poset$ uniquely represents $\poset$, we may consider $H(\poset)$ as a ``shorthand'' for $\poset$ and even loosely say that $H(\poset)$ is a poset.}
\item The \emph{in-degree} (\emph{out-degree}, respectively) of node $u$ of a directed graph $D=(V,E)$ is the number of nodes $v$ such that $vu\in E$ ($uv\in E$, respectively). A directed graph $D$ is an \emph{out-forest} if every node of $D$ has in-degree less than or equal to $1$.
      
      We say  $\poset$ is a \emph{forest} if $H(\poset)$ is an out-forest.
	We say $\poset$ is a \emph{tree} if it is a forest and has a unique maximal element.
	That is, there is a single node in its Hasse diagram of in-degree $0$. 
	
	Note that every forest is a disjoint union of trees.
	Hasse diagrams of a poset, forest and tree are shown in Figure~\ref{fig:hasse-diagrams}.
	The edges in these Hasse diagrams (and all others in the paper) are assumed to be directed from top to bottom.
  \item A set $Y \subseteq X$ is a \emph{chain} if for all distinct pairs of elements $x, y \in Y$, $x <y$ or $y < x$. A chain corresponds to a directed path in $H^*(\poset)$.
  \item A \emph{chain partition} of poset $\poset$ is a disjoint union of chains such that every element of $\poset$ belongs to one of the chains.
  Figure~\ref{fig:hasse-diagram-chain-partition} depicts a chain partition of the poset in Figure~\ref{fig:hasse-diagram-poset}.
  \item Let $x,y \in X$ with $y<x$. 
	Then $\set{z_0, \dots, z_l} \subseteq X$, where $x = z_0 \gtrdot z_1 \gtrdot \dots \gtrdot z_l=y$ is a \emph{derivation chain} (from $x$ to $y$) in $\poset$ of length $l$. 
	A derivation chain from $x$ to $y$ corresponds to a directed path from $x$ to $y$ in $H(\poset)$.
  \item We write $x \shortparallel y$ to indicate that $x,y$ are incomparable, i.e. $x \not\leqslant y$ and $x \not\geqslant y$.  
	A set $Y \subseteq X$ is an \emph{antichain} if for all distinct $x,y \in Y$, $x \shortparallel y$. The \emph{width} of a poset is the cardinality of an antichain of maximum size. 
  \item We write $\dset{x}{}$ to denote $\{y \in X : y \leqslant x\}$ and $\uset{x}{}$ to denote $\{y \in X : y \geqslant x \}$.  
	Note that $\dset{x}{} \subseteq \dset{y}{}$ if and only if $x \leqslant y$. 
\item A {\em linear extension} of $\cal P$ is a chain $(X, \preccurlyeq)$ such that if $x \leqslant y$ then $x \preccurlyeq y$.
  Every (finite) partial order has at least one linear extension, which may be computed, in linear time, by representing the partial order as a directed acyclic graph and using a topological sort~\cite[\S 22.3]{CoLeRiSt09}.
\end{itemize}

 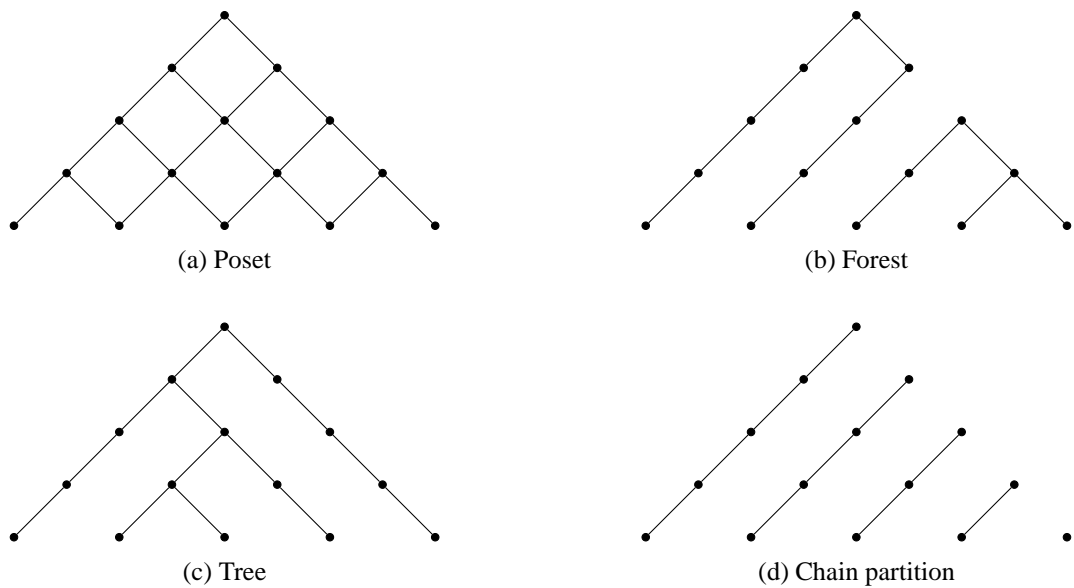
\begin{figure}[h]\centering
  \begin{subfigure}{.475\textwidth}\centering
   \begin{tikzpicture}[vtx/.style={circle,draw,fill,inner sep=1pt},x=1cm,y=.5cm,scale=.7]
    \node[vtx] (11) at (0,0) {};
    \node[vtx] (22) at (2,0) {};
    \node[vtx] (33) at (4,0) {};
    \node[vtx] (44) at (6,0) {};
    \node[vtx] (55) at (8,0) {};
    \node[vtx] (12) at (1,2) {};
    \node[vtx] (23) at (3,2) {};
    \node[vtx] (34) at (5,2) {};
    \node[vtx] (45) at (7,2) {};
    \node[vtx] (123) at (2,4) {};
    \node[vtx] (234) at (4,4) {};
    \node[vtx] (345) at (6,4) {};
    \node[vtx] (1234) at (3,6) {};
    \node[vtx] (2345) at (5,6) {};
    \node[vtx] (12345) at (4,8) {};
    \draw (12345) -- (1234);
    \draw (12345) -- (2345);
    \draw (1234) -- (123);
    \draw (1234) -- (234);
    \draw (2345) -- (234);
    \draw (2345) -- (345);
    \draw (123) -- (12);
    \draw (123) -- (23);
    \draw (234) -- (23);
    \draw (234) -- (34);
    \draw (345) -- (34);
    \draw (345) -- (45);
    \draw (12) -- (11);
    \draw (12) -- (22);
    \draw (23) -- (22);
    \draw (23) -- (33);
    \draw (34) -- (33);
    \draw (34) -- (44);
    \draw (45) -- (44);
    \draw (45) -- (55);
    \end{tikzpicture}
    \caption{Poset}\label{fig:hasse-diagram-poset}
   \end{subfigure}
   \hfill
  \begin{subfigure}{.475\textwidth}\centering
   \begin{tikzpicture}[vtx/.style={circle,draw,fill,inner sep=1pt},x=1cm,y=.5cm,scale=.7]
    \node[vtx] (11) at (0,0) {};
    \node[vtx] (22) at (2,0) {};
    \node[vtx] (33) at (4,0) {};
    \node[vtx] (44) at (6,0) {};
    \node[vtx] (55) at (8,0) {};
    \node[vtx] (12) at (1,2) {};
    \node[vtx] (23) at (3,2) {};
    \node[vtx] (34) at (5,2) {};
    \node[vtx] (45) at (7,2) {};
    \node[vtx] (123) at (2,4) {};
    \node[vtx] (234) at (4,4) {};
    \node[vtx] (345) at (6,4) {};
    \node[vtx] (1234) at (3,6) {};
    \node[vtx] (2345) at (5,6) {};
    \node[vtx] (12345) at (4,8) {};
    \draw (12345) -- (1234);
    \draw (12345) -- (2345);
    \draw (1234) -- (123);
    \draw (2345) -- (234);
    \draw (123) -- (12);
    \draw (234) -- (23);
    \draw (345) -- (34);
    \draw (345) -- (45);
    \draw (12) -- (11);
    \draw (23) -- (22);
    \draw (34) -- (33);
    \draw (45) -- (44);
    \draw (45) -- (55);
    \end{tikzpicture}
    \caption{Forest}\label{fig:hasse-diagram-forest}
   \end{subfigure}
  
   \vspace*{\baselineskip}
   
  \begin{subfigure}{.475\textwidth}\centering
   \begin{tikzpicture}[vtx/.style={circle,draw,fill,inner sep=1pt},x=1cm,y=.5cm,scale=.7]
    \node[vtx] (11) at (0,0) {};
    \node[vtx] (22) at (2,0) {};
    \node[vtx] (33) at (4,0) {};
    \node[vtx] (44) at (6,0) {};
    \node[vtx] (55) at (8,0) {};
    \node[vtx] (12) at (1,2) {};
    \node[vtx] (23) at (3,2) {};
    \node[vtx] (34) at (5,2) {};
    \node[vtx] (45) at (7,2) {};
    \node[vtx] (123) at (2,4) {};
    \node[vtx] (234) at (4,4) {};
    \node[vtx] (345) at (6,4) {};
    \node[vtx] (1234) at (3,6) {};
    \node[vtx] (2345) at (5,6) {};
    \node[vtx] (12345) at (4,8) {};
    \draw (12345) -- (1234);
    \draw (12345) -- (2345);
    \draw (1234) -- (123);
    \draw (1234) -- (234);
    \draw (2345) -- (345);
    \draw (123) -- (12);
    \draw (234) -- (23);
    \draw (234) -- (34);
    \draw (345) -- (45);
    \draw (12) -- (11);
    \draw (23) -- (22);
    \draw (23) -- (33);
    \draw (34) -- (44);
    \draw (45) -- (55);
    \end{tikzpicture}
    \caption{Tree}\label{fig:hasse-diagram-tree}
   \end{subfigure}
   \hfill
  \begin{subfigure}{.475\textwidth}\centering
   \begin{tikzpicture}[vtx/.style={circle,draw,fill,inner sep=1pt},x=1cm,y=.5cm,scale=.7]
    \node[vtx] (11) at (0,0) {};
    \node[vtx] (22) at (2,0) {};
    \node[vtx] (33) at (4,0) {};
    \node[vtx] (44) at (6,0) {};
    \node[vtx] (55) at (8,0) {};
    \node[vtx] (12) at (1,2) {};
    \node[vtx] (23) at (3,2) {};
    \node[vtx] (34) at (5,2) {};
    \node[vtx] (45) at (7,2) {};
    \node[vtx] (123) at (2,4) {};
    \node[vtx] (234) at (4,4) {};
    \node[vtx] (345) at (6,4) {};
    \node[vtx] (1234) at (3,6) {};
    \node[vtx] (2345) at (5,6) {};
    \node[vtx] (12345) at (4,8) {};
    \draw (12345) -- (1234);
    \draw (1234) -- (123);
    \draw (2345) -- (234);
    \draw (123) -- (12);
    \draw (234) -- (23);
    \draw (345) -- (34);
    \draw (12) -- (11);
    \draw (23) -- (22);
    \draw (34) -- (33);
    \draw (45) -- (44);
    \end{tikzpicture}
    \caption{Chain partition}\label{fig:hasse-diagram-chain-partition}
   \end{subfigure}
   \caption{Hasse diagrams of a poset, a forest, a tree, and a chain partition}\label{fig:hasse-diagrams}
 \end{figure}

In many cases we will use subscripts to denote a function or relation relative to a poset $\T$. 
Thus, for example, we write $\T = (X, \leqslantt)$, we write $x \gtrdott y$ if $x >_\mathcal{T} y$ and there is no $z \in X$ such that $x >_\mathcal{T} z >_\mathcal{T} y$, and we write $\dset{x}{\T}$ to denote the set $\{y \in X: y \leqslantt x\}$.

\begin{definition}
An \emph{information flow policy} is a tuple \mbox{$(X,\leqslant,U,O,\lambda)$}, where:
  \begin{itemize}
    \item $(X,\leqslant)$ is a (finite) partially ordered set of \emph{security labels};
    \item $U$ is a set of \emph{users} and $O$ is a set of \emph{objects};
    \item $\lambda: U \cup O \rightarrow X$ is a \emph{security function} that associates users and objects with security labels.
  \end{itemize}
A user $u \in U$ is \emph{authorized} to read an object $o \in O$ if and only if $\lambda(u) \geqslant \lambda(o)$.
\end{definition}

Given an information flow policy $(X,\leqslant,U,O,\lambda)$, we may define an equivalence relation $\sim$ on $U$, where, for any $u, v \in U$, $u \sim v$ if and only if $\lambda(u) = \lambda(v)$.
We write $U_x$ to denote $\set{u \in U : \lambda(u) = x}$. 
Similarly, $O_x \subseteq O$ denotes the set of objects having security label $x \in X$.
In other words, user $u \in U_y$ is authorized to read $o \in O_x$ whenever $y \geqslant x$.
Henceforth, we will represent an information flow policy $(X,\leqslant,U,O,\lambda)$ as a poset $\poset = (X,\leqslant)$ with the tacit understanding that $U$, $O$ and $\lambda$ are given.

\subsection{Cryptographic enforcement}\label{sec:Cryptographicenforcement}

The intuition behind the cryptographic enforcement of information flow policies is to encrypt data objects (using a symmetric encryption algorithm) and distribute appropriate secrets to authorized users (from which encryption keys are derived).
Hence, there are two high-level algorithms that every cryptographic enforcement scheme (CES) provides: the first, \setup, is run by the data owner and generates secrets, keys and any public information that is required for deriving decryption keys; the second, \derive, is used to derive decryption keys from secrets and public information.
That is, in principle, \setup and \derive have the following functionality.
\begin{itemize}
 \item \setup takes as input an information flow policy $(X,\leqslant)$.
 
       \setup outputs $\set{(x,\sigma(x), \kappa(x)): x \in X}$ and $\pub$, where $\sigma(x)$ and $\kappa(x)$ respectively determine the secret and encryption key associated with $x$, and the public information $\pub$ is used as part of the input to the \derive algorithm.
 \item \derive takes as input the information flow policy, $\pub$, $x,y \in X$ and $\sigma(x)$.
       
       \derive outputs $\kappa(y)$ if $y \leqslant x$ (and some distinguished failure symbol $\bot$ otherwise); in particular, $\kappa(x)$ can be derived from $\sigma(x)$.
\end{itemize}

Prior CES schemes follow the above syntactical framework more or less closely. 
In particular, different representations of the information flow policy have been used as input to the \setup and \derive algorithms, and some preprocessing may be required in order to produce those representations.
Some schemes, for example, simply use the Hasse diagram of the poset~\cite{AtBlFaFr09} as the input to \setup and (part of) the input to \derive, while others use a directed, acyclic graph whose edge set is a superset of $\eofh$ and a subset of $\eofhstar$ (and thus contains the same paths as the Hasse diagram)~\cite{AtBlFr07,Cr11}.
In this work, we transform the information flow policy into a partition of trees.

Part of the specification of \derive ensures the \emph{correctness} of a scheme.
That is, an authorized user belonging to $U_x$ must be able to derive $\kappa(y)$ if $x \geqslant y$.
In contrast, the \emph{security} of a CES requires that users cannot derive keys for which they are not authorized, even if they collude by pooling secret information.
In particular, a user in $U_z$ where $z \not\geqslant y$ cannot derive $\kappa(y)$.
Research in the last 10 years, pioneered by~Atallah, Blanton, Frikken and Fazio~\cite{AtBlFaFr09} and Ateniese, de Santis, Ferrara and Masucci~\cite{AtDeFeMa12}, has formalized security notions for CESs.
Informally, the adversary learns the secrets and keys associated with some set of elements $A \subseteq X$ (modeling a group of colluding users) and selects a ``target'' $x$ in $X$ such that $x \not\leqslant a$ for any $a \in A$ (to avoid trivial cases). 
The adversary may be asked to determine $\kappa(x)$ or to determine, given a candidate key $r$, whether $r$ is $\kappa(x)$ or a random element of the key space.
These informal scenarios lead to formal concepts of and definitions for \emph{key recovery} and \emph{key indistinguishability}~\cite{AtBlFaFr09}.%
\footnote{Note that a scheme in which \derive may be used to compute $\kappa(y)$ from $\kappa(x)$ whenever $y < x$ (rather than from $\sigma(x)$) does not possess the key indistinguishability property: the adversary may select $x$ and $A$ such that $x \gtrdot a$ for some $a \in A$, use $x$, $a$, $\pub$ and $r$ (that is, assume $r = \kappa(x)$) as inputs to the \derive algorithm, and test the output for equality with $\kappa(a)$.  
	  Concerns about key indistinguishability in CESs led to the separation between secrets and keys~\cite{AtBlFaFr09}.}
We consider the security properties of CESs in more detail in Section~\ref{sec:security-analysis}.

\subsection{Related Work}\label{sec:related-work}

Essentially, designing a cryptographic enforcement scheme comes down to defining%
\begin{inparaenum}[(i)]
 \item what \emph{secrets} each user will receive, 
 \item how users will generate any \emph{keys} they require to decrypt data objects, and 
 \item how secrets and keys are related.
\end{inparaenum}
Broadly speaking, there are two standard ways of designing a cryptographic enforcement scheme for information flow policies.
These methods assume each user is given a single key from which all other relevant secrets and key may be derived, and are distinguished by the information used to derive secrets and keys.
The first method, which we will call ``node-based'', relies only on secret information known to the user, while the second, which we will call ``edge-based'', assumes that some additional information must be made known to all users.%
\footnote{There are some other types of schemes but each of them suffer from a number of disadvantages (see~\cite{CrMaWi06}, for example) so research has tended to focus on node- and edge-based schemes.}

Informally, a node-based scheme uses one-way functions: for $y < x$ the secret associated with $y$ is some (one-way) function of $\sigma(x)$, the secret $\sigma(x)$ associated with $x$, and $\kappa(y)$, the key associated with $y$, is some (one-way) function of $\sigma(y)$.
Some of the earliest work on cryptographic enforcement of information flow policies used these kinds of techniques~\cite{Sa88}.
However, in this setting, it is unclear how to distribute secrets such that $\sigma(y)$ can be derived from $\sigma(x_i)$ for each of the parents $x_1,\dots,x_n$ that node $y$ might have, without simultaneously exposing the scheme to collusion attacks.

In an edge-based scheme, public information is associated with each pair $(x,y)$ where $x > y$ from which $\sigma(y)$ can be extracted with knowledge of $\sigma(x)$.
Thus, informally, we might define $\pub(x,y)$ to be $\mathrm{enc}_k(\sigma(y))$, where $\mathrm{enc}_k$ is some symmetric encryption algorithm with key~$k$ contained in $\sigma(x)$.
An edge-based scheme can be used for arbitrary posets but requires public information~\cite{AtBlFaFr09}.

Research into schemes that allocate a single secret to each user investigated what trade-offs were possible between the number of items of public data and the number of key derivation operations (in the worst case)~\cite{AtBlFr06,Cr11}.
Some of this work focused on posets with a particular structure (such as chains~\cite{AtBlFr06}).
Such research was able to define specific data structures and algorithms, and perform exact complexity analyses~\cite{AtBlFr06,AtBlFr07,Cr11}.
Other work considered arbitrary posets and used results from graph and poset theory to develop analyses that were generic but arguably less useful in specific cases~\cite{AtDeFeMa12}.
In all this work, the amount of public information required for key derivation necessarily increases.

A representation of the policy is required as input to the \derive\ algorithm.
Hence, the data owner must publish the policy (or distribute it with the appropriate secrets to every user).
The size of the policy is proportional to the number of edges (each representing a piece of public information) used for secret derivation; that is $O(n^2)$, where $n$ is the cardinality of $X$ (the set of security labels).
However, compact representations, using an $n \times n$ binary matrix, exist.
In the case of edge-based schemes, the data owner must also publish (or otherwise distribute) $\pub$, which is also proportional in size to the number of edges.
However, the size of $\pub$ will be several orders of magnitude bigger than the policy representation (due to the relative sizes of each datum of information).
An alternative is to store $\pub$ on a public server.
In this case, the server must be on-line and accessible to any user that wishes to run the \derive\ algorithm.
Thus, it may be advantageous to devise schemes that require no public information.

Crampton \emph{et al.}~\cite{CrDaMa10} introduced the idea of cryptographic enforcement schemes, based on chain partitions of the information flow policy, that require no public information.
The trade-off with such schemes is that some users may require more than one secret in order to be able to derive all the required encryption keys.
Subsequent work established that secure instantiations of such schemes are possible~\cite{FrPa11,FrPaPo13}.

To summarize, informally, the core trade-off made when designing a CES is the amount of public information that is required to assist in the derivation of secrets against the number of additional secrets that are associated with nodes.
Broadly speaking, on the one hand one assumes each node is associated with a single secret and defines a ``secret-derivation digraph'' $G = (X,E)$, where $E_{\min} \subseteq E \subseteq E_{\max}$.
(In other words, if $x > z$ in $(X,\leqslant)$ there is a derivation path in $G$, since $E \supseteq E_{\min}$; and if $x \ngtr z$ there is no derivation path in $G$, since $E \subseteq E_{\max}$.)
On the other hand, one selects a secret-derivation digraph $G =(X,E)$ such that $E \subset E_{\min}$, $G$ is an out-forest, and each node is associated with at least one secret.
Then, if $x > z$, there is some node $y$ such that every user in $U_x$ is given the secret  associated with node $y$ and there is a directed path from $y$ to $z$ in $G$.
Figure~\ref{tbl:scheme-comparison} provides a crude comparison of the generic schemes in the literature: $E$ is the set of edges used to derive secrets; $d$ is the length of the longest directed path in $G = (X,E)$; $w$ is the width of $X$; $n$ is the cardinality of $X$. 

\begin{figure}[h]\centering
 \begin{tabularx}{\textwidth}{llXXX}
 \toprule
  \bf Generic scheme & \bf Edge set & \bf Public \mbox{information} & \bf Derivation time & \raggedright\arraybackslash \bf Secrets per node \\
 \midrule
  Single-step secret derivation & $E = E_{\max}$ & $O(\card{E})$ & $O(1)$ & $k = 1$ \\
  Multi-step  secret derivation & $E_{\min} \subseteq E \subset E_{\max}$ & $O(\card{E})$  & $O(d)$ & $k = 1$\\
  Chain-based  secret derivation & $E \subset E_{\min}$ & None & $O(d)$ & $k \in [1,w]$ \\
  All secrets distributed & $E = \emptyset$ & None & $0$ & $k \in [1,n]$ \\
 \bottomrule
 \end{tabularx}
 \caption{A high-level comparison of generic cryptographic enforcement schemes}\label{tbl:scheme-comparison}
\end{figure}

The significant open problem with prior work on chain-based schemes is the assumption that the chain partition is part of the input to the \setup algorithm: there may be many such partitions and it is not immediately obvious how one should select a specific partition in order to optimize characteristics of the corresponding enforcement scheme (an example being to minimize the number of secrets issued).
Hence, it seems very natural to ask how difficult it is to compute a ``good'' chain partition, given that %
\begin{inparaenum}[(i)]
 \item schemes based on chain partitions do not require public information, and 
 \item the number of secrets that need to be distributed to users is determined by the choice of chain partition. 
\end{inparaenum}
Our recent work~\cite{CrFaGuJo15a} shows that it is possible to compute a minimal chain partition in polynomial time using a minimum cost network flow algorithm.

Crampton \emph{et al.}~\cite{CrFaGuJoPo14} made use of the fact that derivation paths are uniquely defined in trees (as well as in chains) to develop the idea of a tree-based cryptographic enforcement scheme.  
Their work established that it was possible to compute (in polynomial time) an optimal tree for the information flow policy. 
  
\subsection{Problem overview}

While chain-based enforcement schemes require no public information, some users may be required to store more than one {secret}, unlike the majority of schemes in the literature.
The number of secrets required by an instantiation of such a scheme depends on the chain partition chosen.
Moreover, a natural extension of the chain-based approach, explored in the current work, is to use a forest related to the poset defining the information flow policy.
In this paper, therefore, we explore three questions:
\begin{itemize}
 \item What is the optimal choice of chain partition and can we compute such a partition efficiently?
 \item How do we implement a cryptographic enforcement scheme based on a  partition of the information flow policy into trees rather than chains?
 \item What is the optimal choice of tree partition and can we compute such a partition efficiently?
\end{itemize}
In the next section, we consider the second of these questions, the results of which enable us to answer the other two questions.

\section{Enforcement Schemes from Tree Partitions}\label{sec:schemes}
 
 In this section, we generalize the approach taken by Crampton \emph{et al.}~\cite{CrDaMa10} for chain-based enforcement schemes, and Crampton \emph{et al.}~\cite{CrFaGuJoPo14} for tree-based enforcement schemes.
 In particular, we introduce the concept of a tree partition of a poset $(X,\leqslant)$ and show how such a partition may be used to construct a cryptographic enforcement scheme for an information flow policy defined by $(X,\leqslant)$.
 
 \begin{definition}
  Let $\poset = (X,\leqslant)$ be a poset, with Hasse diagram $H(\poset) = (X,E)$.
  A \emph{tree partition} of $\poset$ is a poset $\T = (X, \leqslantt)$ such that $H(\T) = (X,E_{\T})$ is an out-forest and $E_{\T} \subseteq E$.
 \end{definition}

 If $\poset = (X, \leqslant)$ is a poset, $\T = (X,\leqslantt)$ is a tree partition of $\poset$ and $y \nless x$, then $y \nlesst x$.
 However, we may have $y < x$ but $y \nlesst x$.
 Thus, the problem with a tree partition, in the context of cryptographic enforcement schemes (CESs), is that some authorized labels that were ``reachable'' by a derivation chain in $\poset$ will no longer be reachable in $\T$.
 Accordingly, we define the notion of forest-based enforcement scheme for a tree partition of $\poset = (X,\leqslant)$.

 \begin{definition}
  Given an information flow policy $\poset = (X,\leqslant)$ and a tree partition $\T = (X,\leqslantt)$, a \emph{forest-based enforcement scheme} is a pair $(\T,\generickaf)$, where $\generickaf : X \rightarrow 2^X$ and:
  \begin{enumerate}
   \item if $u \leqslant x$ then there exists $z \in \generickaf(x)$ such that $u \leqslantt z$;
   \item if $u \not\leqslant x$ then for all $z \in \generickaf(x)$, $u \not\leqslantt z$.
  \end{enumerate}
 \end{definition}

 Informally, conditions 1 and 2 correspond to the correctness and security requirements of CESs, respectively.
 Note that $x \in \generickaf(x)$.
 To see this, suppose, in order to obtain a contradiction, that $x \not\in \generickaf(x)$.
 Then, by the first property, there exists $z \in \generickaf(x)$ such that $x \lesst z$. 
 This implies $x < z$ and thus $z \not\leqslant x$. 
 By the second property for all $z^* \in \generickaf(x)$ we then have $z \not\leqslantt z^*$. 
 This holds in particular for $z^*=z$ and we obtain $z \not\leqslantt z$, a contradiction.
 
 \begin{definition}
  Let $\poset$ be a poset and $\T$ a tree partition of $\poset$.
  Then, given $x,z \in X$, the maximum element (if it exists) in $\dset{x}{\poset} \cap \uset{z}{\T}$, is the \emph{anchor} between $x$ and $z$ and denoted by $\bridge{xz}$.  
 \end{definition}
 
 We note the following facts, which we state without proof:
 \begin{itemize}
  \item $\bridge{xz}$ exists iff $x \geqslant z$;
  \item $\bridge{xz}$ is a unique maximal element (that is, a maximum element) since $\uset{z}{\T}$ is a chain;
  \item if $x \geqslant z$ and $x \gtrt z$ then there exists a derivation chain in $\T$ from $x$ to $z$ and $\bridge{xz} = x$ (since $x$ is the maximum element in $\dset{x}{\poset}$); and
  \item if $x \geqslant z$ and $x \ngtrt z$ then there exists a derivation chain in $\T$ from $\bridge{xz}$ to $z$ and $x > \bridge{xz}$.
 \end{itemize}

 Given $\poset$ and a tree partition $\T$, define $\bestkaf{\T} : X \rightarrow 2^{X}$ as follows: 
 \[
  \bestkaf{\T}(x) = \set{\bridge{xz} : x \geqslant z}
 \]

 \begin{proposition}
  For any poset $\poset$ and any tree partition $\T$ of $\poset$, $(\T,\bestkaf{\T})$ is a forest-based enforcement scheme.
 \end{proposition}
 
 \begin{proof}
  If $u \leqslant x$, then $z = \bridge{xu}$ belongs to $\bestkaf{\T}(x)$ and $u \leqslantt z$.
  And if $u \not\leqslant x$ then for every $z\in \phi_{\T}(x)$ we have $z\not\geqslantt u$.
 \end{proof}
 
  In other words, given the secrets corresponding to the elements in $\bestkaf{\T}(x)$, a user in $U_x$ can derive the secret for all elements $z \leqslant x$ using a derivation chain starting at $\bridge{xz}$.

 \begin{lemma}\label{lem:phi-f-is-best-enforcement-scheme}
  Let $\poset = (X,\leqslant)$ be a poset, $\T$ be a tree partition of $\poset$, and $(\T,\psi)$ be a forest-based enforcement scheme.
  Then \mbox{$\bestkaf{\T}(x) \subseteq \psi(x)$} for all $x \in X$.  
 \end{lemma}
 
 \begin{proof}
  Suppose, in order to obtain a contradiction, that $y \in \bestkaf{\T}(x)$ and $y \not\in \psi(x)$.
  By definition, $y \leqslant x$; therefore, there must exist $y' \in \psi(x)$ such that $y' \gtrt y$, and thus $x \geqslant y'$.
  Moreover, $y \geqslantt z$ so we have $x \geqslant y' \gtrt y \gtrt z$; that is, $y' \in \dset{x}{\poset} \cap \uset{z}{\T}$.
  Thus $y$ is not the maximal element in $\dset{x}{\poset} \cap \uset{z}{\T}$, the desired contradiction.
 \end{proof}
 
 The following simple lemma characterizes the elements of $\bestkaf{\T}$ and will be used to prove Proposition \ref{pro:computation-time-for-phi} and Theorem \ref{thm:building-optimal-tree-partition}.
 
 \begin{lemma}\label{lem:phi-x-equals-gamma-y-z}
   Let $\T = (X,\leqslantt)$ be a tree partition of poset $\poset = (X,\leqslant)$.
   Then for every $x$ in $X$ and every $z$ in $X$, $z \in \bestkaf{\T}(x)$ if and only if exactly one of the following conditions holds:%
    \begin{inparaenum}[(i)] 
      \item $z = x$;
      \item $z<x$, $z$~has a parent in $\T$ and $x \not\geqslant \parentt{z}$;
      \item $z < x$ and $z$ has no parent  in ${\cal T}$.
    \end{inparaenum}
 \end{lemma}

 \begin{proof} 
 Suppose 
  $x \geqslant z$ and $x \not\geqslant \parentt{z}$. Since $x \not\geqslant \parentt{z} \gtrdot_{\cal T} z$,
  $z$ is the maximal element in $\dset{x}{\poset} \cap \uset{z}{\T}$.
  Similarly, if $z$ has no parent or $z=x$, then $z$ is the maximal element in $\dset{x}{\poset} \cap \uset{z}{\T}$.
  In either case, $z = \bridge{xz}$ and $z \in \bestkaf{\T}(x)$.
  
  Conversely, if $z \in \bestkaf{\T}(x)$, then $x \geqslant z$, by definition, and $\bridge{xz} = z$.
  Thus, $x \not\geqslant \parentt{z}$ if $z$ has a parent (otherwise, $\parentt{z} \in \dset{x}{\poset} \cap \uset{z}{\T}$ and $z \ne \bridge{xz}$).
 \end{proof}


 
 \begin{proposition}\label{pro:computation-time-for-phi}
  Let $\poset=(X,\leqslant)$ be an information flow policy and let $\T=(X,\mathop{\leqslantt})$ be a tree partition.
  Then $\bestkaf{\T}$ can be computed in time $O(n^2)$, where $n = \card{X}$.
 \end{proposition}
 
 \begin{proof}  
%
  By Lemma~\ref{lem:phi-x-equals-gamma-y-z}, for all $x \in X$, besides $x$ itself, we add all those elements $z\in X$, $z<x$, to $\bestkaf{\T}(x)$ that are either maximal in $\T$ or, if not, satisfy $x \not\geqslant \parentt{z}$.
  In both cases, we must determine whether $x > z$ for some $z \in X$.
  
After $O(n^2)$ time preprocessing, we may assume that we have data structures allowing us to check whether $x > z$ in $O(1)$ time, and test whether $z$  is a maximal element in $T$ (and compute $\parentt{z}$ otherwise) in $O(1)$ time. 
  Hence, we can compute $\bestkaf{\T}$ in $O(n^2)$ time.
 \end{proof}

\subsection{Generic instantiation}\label{subsec:generic-instantiation}

The above results enable us to specify the algorithms of a cryptographic enforcement scheme.
The construction can be considered a generalization of the one using chains (rather than trees) defined by Freire \emph{et al.}~\cite{FrPaPo13}.
When defining $\setup$ and $\derive$ we assume that the information flow policy $\poset=(X,\leqslant)$ is presented in the form of a tree partition $\T = (X,\leqslantt)$,
and that for the latter a specific forest-based enforcement scheme $(\T,\psi)$ has been selected
(such as $(\T,\bestkaf{\T})$).
Further, for the $n=\lvert X\rvert$ labels of~$X$ we assume a numbering convention that follows a (reverse) linear extension $\prec$ of~$\leqslant$;
more precisely, we assume that $X=\{x_1,\ldots,x_n\}$ where $x_n \prec x_{n-1} \prec \dots \prec x_2 \prec x_1$ 
(in particular, $x_n$ is a minimal element in $X$ and $x_1$ is a maximal element).
The cryptographic building block of our construction is a pseudorandom function (PRF) where the key space and the output space are the same set~$\keysp$.
Given such a function $\prf\colon\keysp\times\set{0,1}^*\to\keysp$ and an (injective) label naming function $\ell\colon X\to\set{0,1}^*$ we define:

\bigskip\noindent Algorithm~$\setup$, on input an information flow policy in the format described above:
\begin{enumerate}
 \item For $i=1$ to $n$ do (i.e., count from a maximal down to a minimal label):
\label{line:setupprfderivation}
   \begin{itemize}
   \item if $x_i$ is maximal in $(X,\leqslantt)$ pick fresh random key $s(x_i)\getsr\keysp$;
   \item otherwise, identify the (unique) parent $y$ of $x_i$ in $\T$ and assign $s(x_i)\gets\prf(s(y),\ell(x_i))$ (where $s(y)$ is the PRF key and $\ell(x_i)$ is the PRF input);
   \end{itemize}
 \item For each $x \in X$ output $\sigma(x)=\{(v,s(v)):v\in\generickaf(x)\}$ and $\kappa(x) = \prf(s(x),\ell(x))$; 
   no public information is needed, i.e., $\pub=\emptyset$.
\end{enumerate}
The general principle of this CES is to derive secrets in a top-down fashion:
top nodes (according to $\leqslantt$) are assigned random keys, 
and the keys of all other nodes are deterministically derived from their parent using the PRF.
Observe that, as we arranged $\prec$ to be a linear extension of $\leqslant$ (and thus $\leqslantt$), step~\mbox{(\ref{line:setupprfderivation})} of $\setup$ is actually well-defined.
We next define the corresponding $\derive$ algorithm:

\bigskip\noindent Algorithm~$\derive$, on input the information flow policy, labels $x,y\in X$, and secret $\sigma(x)$:
\begin{enumerate}
 \item Return $\bot$ if $x \not\geqslant y$;
 \item Identify the (unique) $z\in\generickaf(x)$ such that $y \leqslantt z$ and recover $s(z)$ from $\sigma(x)$;
 \item Let $z=z_0 \gtrdot z_1 \gtrdot \dots \gtrdot z_m=y$ be the complete derivation chain in $\T$ between $z$ and $y$;
 \item For $i=1$ to $m$ do: $s(z_i)\gets\prf(s(z_{i-1}),\ell(z_i))$;
 \item Output $\kappa(y) = \prf(s(y),\ell(y))$.
\end{enumerate}

In this instantiation, the same pseudorandom function $\prf$ is used as a secret- and key-generation function;
secret values, and values derived from secret values, serve as PRF keys, and fixed strings that uniquely identify the corresponding node are its inputs.

\subsection{Security analysis}\label{sec:security-analysis}

We assess the security of our enforcement scheme using the principles of provable security.
We start by formalizing the properties of the cryptographic building block, the pseudorandom function~$\prf$.
Our definition is not the most general possible: rather, it is tailored to the requirements of our construction; 
specifically, we require that the keyspace and the range of the PRF are the same set.

\begin{definition}
  \label{def:prfadvantage}
  A \emph{pseudorandom function} (PRF) with keyspace and range~$\keysp$ is any efficient function $\prf\colon\keysp\times\set{0,1}^*\to\keysp$.
  We also write $\prf_{K}(x)$ to denote $\prf(K,x)$.
  We define the \emph{advantage} of an adversary~$\cD$ in distinguishing $\prf$ from a random function as   
  \[
  \Adv^\prf(\cD)=\left\lvert\Pr[K\getsr\keysp;\cD^{\prf_K}\Rightarrow 1]-\Pr[\varphi\getsr\langle\set{0,1}^*\to\keysp\rangle;\cD^{\varphi}\Rightarrow 1]\right\rvert\, .
  \]
  We say that PRF $\prf$ is \emph{$(\epsilon,\tau)$-indistinguishable} from a random function if $\epsilon$ upper-bounds the advantage of all distinguishers~$\cD$ that run in time at most~$\tau$.
\end{definition}

In the definition above, $\langle\set{0,1}^*\to\keysp\rangle$ denotes the universe of all functions mapping $\set{0,1}^*$ to~$\keysp$,
and writing ``$\cD^F\Rightarrow 1$'' for a function~$F$ means that algorithm~$\cD$ has oracle access to $F$ and terminates outputting value~$1$.
In Definition~\ref{def:prfadvantage}, $F$~either implements access to a keyed PRF instance $\prf_K$, or it implements a completely random function.
That is, the smaller we can choose~$\epsilon$, the closer a particular PRF $\prf$ is to a random function.
We discuss some practical candidate functions in Section~\ref{sec:PRFinpractice}.

We next make precise the level of security that we target for our enforcement scheme.
Many different cryptographic models for CES with security guarantees of various strengths have been proposed
(see~\cite{CaSaMa14} for a comparative overview).
The notion we target and reproduce below, strong key indistinguishability~\cite{FrPaPo13},
was not only proven to imply all other notions (i.e., to define the highest level of security),%
\footnote{\cite{CaSaMa14} show that not all of these implications are strict;
in particular strong key indistinguishability is polynomially equivalent to the notion of (plain) key indistinguishability of~\cite{AtBlFaFr09}, with tightness loss $n=\lvert X\rvert$.
Note also our model considers a static setup where the challenge label is fixed a~priori.
    A variant of Definition~\ref{def:kist} would consider dynamic adversaries: such an adversary is able to choose the challenge label $x$ \emph{during} the experiment, rather than having it fixed as one of the experiment's parameters. 
	    However, it has been shown that static and dynamic definitions of strong key indistinguishability are polynomially equivalent~\cite{FrPaPo13};
            corresponding results for (plain) key indistinguishability have also been obtained~\cite{AtDeFeMa12}.
	    To simplify the exposition, therefore, we restrict our attention to the static case.}
but is also, we believe, the most natural and versatile one.
It is based on the security experiment $\Expt^{\kist,b}_{X,x}$ defined in Fig.~\ref{fig:kist}, where we use the following notation:
  \begin{align*}
   \bar{\sigma}&=\set{(v,\sigma(v)): v \in X}, \\
   \bar{\kappa}&=\set{(v,\kappa(v)): v \in X}, \\
   \ExpCorrupt_{X,x}&=\set{(v,\sigma(v)): v\in X, x\not\leqslant v}, \\
   \ExpKeys_{X,x}&=\set{(v,\kappa(v)):v\in X\setminus\{x\}}.   
  \end{align*}
In the experiment we assume that the adversary receives the information flow policy~$(X, \leqslant)$ in the same format as the $\setup$ algorithm does.

\begin{figure}[ht]
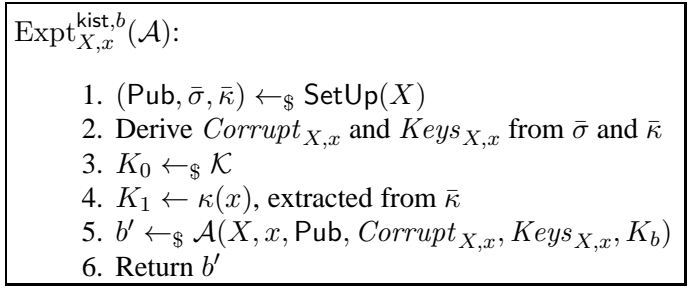

 \newlength{\mywidth}
 \settowidth{\mywidth}{\qquad\qquad Derive $\ExpCorrupt_{X,x}$ and $\ExpKeys_{X,x}$ from $\bar{\sigma}$ and $\bar{\kappa}$}
  \centering
    \fbox{\begin{minipage}[t]{\mywidth}
      $\Expt^{\kist,b}_{X,x}(\cA)$: \\[1mm]\null\quad
    \begin{minipage}[t]{\mywidth}
    \begin{enumerate}
    \item $(\pub,\bar{\sigma},\bar{\kappa})\getsr\setup(X)$
    \item Derive $\ExpCorrupt_{X,x}$ and $\ExpKeys_{X,x}$ from $\bar{\sigma}$ and $\bar{\kappa}$\label{lin:expderivekeys}
    \item $K_0\getsr\keysp$\label{lin:K0}
    \item $K_1\gets \kappa(x)$, extracted from $\bar{\kappa}$\label{lin:K1}
    \item $b'\getsr\cA(X,x,\pub,\ExpCorrupt_{X,x},\ExpKeys_{X,x},K_b)$\label{lin:expcalladv}
    \item Return $b'$
    \end{enumerate}
  \end{minipage}%
  \end{minipage}}%
  \caption{Security experiment for strong key indistinguishability}
  \label{fig:kist}
\end{figure}

\begin{definition}
  \label{def:kist}
  Let $(X,\leqslant)$ be an arbitrary poset.
  A CES for $(X,\leqslant)$ is \emph{$(\epsilon,\tau)$-strongly key indistinguishable with respect to static adversaries}~\cite{FrPaPo13} if, for all $x\in X$, the advantage of all adversaries~$\cA$ that interact in experiment $\Expt^{\kist,b}_{X,x}(\cA)$ and run in time at most~$\tau$ is bounded by~$\epsilon$, where we define
  \[
  \Adv^{\kist}_{X,x}(\cA)=\left\lvert\Pr\left[\Expt^{\kist,1}_{X,x}(\cA)\Rightarrow1\right]-\Pr\left[\Expt^{\kist,0}_{X,x}(\cA)\Rightarrow1\right]\right\rvert.
  \]
\end{definition}

Observe that in this definition the adversary obtains, in principle, all secrets embedded in the system (that is, all $\sigma(x)$ and $\kappa(x)$ values), excluding only those that would allow distinguishing the challenge key by trivial means (e.g., by invoking the $\derive$ algorithm).

The final step of our analysis is to prove that our forest-based enforcement scheme from the preceding section is strongly key indistinguishable in the sense of Definition~\ref{def:kist}.
More precisely, we have the following result.

\begin{theorem}\label{thm:security}
  For any poset $(X,\leqslant)$, $x\in X$, and adversary~$\cA$ that runs in time at most~$\tau$, there exists a constant $0\leqslant c\leqslant\lvert X\rvert$ and distinguishers $\cD^0_1,\ldots,\cD^0_c$, $\cD^1_1,\ldots,\cD^1_c$ against the underlying PRF such that
  \[
  \Adv^{\kist}_{X,x}(\cA) \quad\leqslant\quad \Adv^\prf(\cD^0_1)+\dots+\Adv^\prf(\cD^0_c)+\Adv^\prf(\cD^1_1)+\dots+\Adv^\prf(\cD^1_c)
  \]
  and the respective running times are at most $\tau^b_i=\tau+O(\lvert X\rvert)$.
  That is, if the PRF is $(\epsilon',\tau+O(\lvert X\rvert))$-indistinguishable then our CES construction is $(\epsilon,\tau)$-strongly key indistinguishable with $\epsilon=2\lvert X\rvert \epsilon'$.
\end{theorem}
 
\begin{proof}
  The argument proceeds using a sequence of $\lvert X\rvert = n$ hybrid games that interpolate between experiments $\Expt^{\kist,0}_{X,x}$ and $\Expt^{\kist,1}_{X,x}$.
  In each hybrid step, if specific conditions are met, we replace one PRF instance by a random function;
  from the point of view of the adversary, the distance between each two consecutive hybrids is not greater than $\Adv^\prf(\cD)$ for a specific PRF distinguisher~$\cD$.

  Fix a poset $(X,\leqslant)$ together with a (reverse) linear extension $x_n \prec x_{n-1} \prec \dots \prec x_2 \prec x_1$ of~$X$,
  a label $x\in X$, and a CES adversary~$\cA$ that runs in time at most~$\tau$.
  We use sequence $x_n \prec \dots \prec x_1$ to define our hybrid experiments:
  For $b\in\set{0,1}$, we set $G^b_0=\Expt^{\kist,b}_{X,x}$ and define games $G^b_1,\ldots,G^b_n$ (in that order) such that if $1\leq k\leq n$ and $x_k\geqslant x$ then the difference between games~$G^b_k$ and~$G^b_{k-1}$ is precisely that all PRF invocations with key $s(x_k)$ are replaced by assignments with values drawn uniformly at random from~$\keysp$
  (correspondingly, also the keys considered in lines~(\ref{lin:expderivekeys}) and~(\ref{lin:K1}) are changed).
  For the remaining indices~$k$, i.e., in case $x_k\not\geqslant x$, games $G^b_k$ and~$G^b_{k-1}$ are identical.
  Let $S^b_k$ denote $\Pr[G^b_k(\cA)\Rightarrow1]$ for all~$b,k$.

  Observe that we replace PRF invocations by random assignments for precisely those labels $x\in X$ that do not have a corresponding entry in $\ExpCorrupt_{X,x}$.
  Observe also that, as we consider the labels in a suitable order, for all switchings from a PRF to a random function we have that the corresponding PRF key $s(x)$ was replaced with a uniform random value before.
  Thus, the difference between any two consecutive games is bounded by a PRF advantage: 
  by a standard reductionist argument, in the cases $x\leqslant x_k$, we have 
   \begin{equation}\label{eq:diff-games}
     \lvert S^b_k-S^b_{k-1}\rvert=
     \lvert\Pr[G^b_k(\cA)\Rightarrow 1]-\Pr[G^b_{k-1}(\cA)\Rightarrow 1]\rvert=\Adv^\prf(\cD)
     \enspace,     
   \end{equation}
  for a specific distinguisher~$\cD$ with running time approximately $\tau+\lvert X\rvert\cdot T_{\rm prf}\in\tau+O(\lvert X\rvert)$,
  where $T_{\rm prf}$ is the time required for one PRF evaluation;
  in addition, whenever $x\not\leqslant x_k$ we have $G^b_k=G^b_{k-1}$ and hence $\lvert S^b_k-S^b_{k-1}\rvert=0$.
  Now, by repeated application of the triangle inequality and~\eqref{eq:diff-games}, we have
  \[
    \left\lvert S^b_0-S^b_n\right\rvert
    \leqslant \sum_{k=1}^n \left\lvert S^b_{k-1}-S^b_{k}\right\rvert
       \leqslant \sum_{k=1}^c \Adv^\prf(\cD^b_k)
  \enspace,
  \]
  where $c=\lvert\{x'\in X:x\leqslant x'\}\rvert$ and distinguishers $\cD^b_k$ are constructed as specified.
  We now consider games $G^0_n$ and $G^1_n$.  
  In both cases $\kappa(x)$ is picked uniformly at random, thus lines~(\ref{lin:K0}) and~(\ref{lin:K1}) in the experiment implement the same operation.  
  Hence $G^0_n$ is identical to $G^1_n$ and $\left\lvert S^0_n-S^1_n\right\rvert=0$. 
  Thus, we obtain
  \begin{eqnarray*}
  \Adv^{\kist}_{X,x}(\cA)&=&\lvert S^1_0-S^0_0\rvert\leqslant
  \lvert S^1_0-S^1_n\rvert+\lvert S^1_n-S^0_n\rvert+\lvert S^0_n-S^0_0\rvert\\
  & \leqslant & \Adv^\prf(\cD^1_1)+\ldots+\Adv^\prf(\cD^1_c) + 0 + \Adv^\prf(\cD^0_1)+\ldots+\Adv^\prf(\cD^0_c)
  \end{eqnarray*}
  as required.
\end{proof}

Note that by results of~\cite{CaSaMa14} it would have sufficed to prove (plain) key indistinguishability of our scheme, as the latter would imply the notion of strong key indistinguishability that we target.
Observe however that going this way introduces a tightness loss of $n=\lvert X\rvert$.
Besides saving this factor, we believe our direct approach is also more intuitive.

\subsection{On practical instantiations of the PRF component}
\label{sec:PRFinpractice}

We now briefly consider how one might instantiate our CES in practice.
Although pseudorandom functions are a standard building block in the domain of provable security, 
corresponding constructions do not explicitly appear in most international cryptographic standards documents (e.g., by ANSI, IEEE, NIST, IETF, etc.).
However, certain standardized MACs and block ciphers can be used as a PRF replacement, as we discuss next.

The primary aim of message authentication codes (MACs) is integrity protection and data authentication.
A standard result says that any PRF may also be used as a MAC.
The converse is in general not true: 
a good MAC is not automatically a good PRF.
Fortunately, however, essentially all standardized MAC constructions are in fact good PRFs,
including the popular HMAC~\cite{FIPS:198-1}, CMAC~\cite{NIST:SP800-38B}, GMAC~\cite{NIST:SP800-38D}, and PMAC~\cite{C:BlaRog02} schemes.

In our application, the data input of the PRF and hence of the MAC is the name $\ell(x)$ of a node $x\in X$.
For the sake of generality we did not impose any constraints on the format of these names (in particular, strings of arbitrary length are allowed).
We note that all of the MAC schemes mentioned above are designed to process arbitrary-length strings, of any format.
By consequence, all of them are suitable to securely instantiate our enforcement scheme.
However, we point out that if we imposed a constant-length restriction on $\ell(x)$, then a much simpler PRF than the MACs mentioned above can be used:
by the PRF/PRP switching lemma \cite{C:BlaRog00}, any block cipher (a.k.a.\ pseudorandom permutation, PRP) also constitutes a PRF, where the input length is equal to the output length and coincides with the cipher's block size.
In particular, if one is satisfied with using 128~bit keys and may require 128-bit labels for elements in~$X$ 
then the AES block cipher can be used without modification as the pseudorandom function of our CES construction.
Further, if the target is a security level of 256~bit and one uses 127-bit labels, then the following function would be a suitable PRF:
\[\prf\colon\set{0,1}^{256}\times\set{0,1}^{127}\to\set{0,1}^{256},\ \text{where}\ (K,s)\mapsto {\mathrm{AES}}_K(0\parallel s)\parallel{\mathrm{AES}}_K(1\parallel s) .\]

\section{Selecting a Good Tree Partition}\label{sec:tree-based-schemes}

 Each poset admits many possible tree partitions and each tree partition gives rise to many possible enforcement schemes.
 In this section, we investigate which enforcement scheme to select for a given tree partition and which tree partition to select for a given poset.
 Our analysis is based on the assumption that we wish to minimize the total number of secrets that need to be distributed to users.
 Thus, given a tree partition $\T = (X,\leqslantt)$ and a forest-based enforcement scheme $(\T,\generickaf)$, we define
 \[
  \numsecrets{\generickaf}{\T} = \sum_{x \in X} \card{\generickaf(x)} \cdot \card{U_x}.
 \]
 Note that $\card{\generickaf(x)}$ denotes the number of secrets issued to each $u \in U_x$ for the enforcement scheme $(\T,\generickaf)$.
 Thus, $\numsecrets{\generickaf}{\T}$ is the total number of secrets that need to be distributed to users when we apply scheme $(\T, \generickaf)$.
 By Lemma~\ref{lem:phi-f-is-best-enforcement-scheme}, for a given tree partition $\T = (X,\leqslantt)$, any forest-based enforcement scheme $(\T,\generickaf)$ and any $x \in X$, we have $\bestkaf{\T}(x) \subseteq \generickaf(x)$; thus
    $\card{\bestkaf{\T}(x)} \leqslant \card{\generickaf(x)}$ and \mbox{$\numsecrets{\bestkaf{\T}}{\T} \leqslant \numsecrets{\generickaf}{\T}$}.
 Hence, for a given tree partition $\T$, we will assume the use of the forest-based enforcement scheme $(\T,\bestkaf{\T})$.

 Let $\poset = (X, \leqslant)$ be an information flow policy and let $\T = (X, \leqslantt)$ be a tree partition of $\poset$.
 Then we say that $\T$ is a \emph{minimal tree partition} of $\poset$ if, for any tree partition $\T'$ of $\poset$, we have 
 $\numsecrets{\bestkaf{\T}}{\T} \leqslant \numsecrets{\bestkaf{\T'}}{\T'}$.
 (In other words, $\T$ is a tree partition that minimizes the total number of distributed secrets.)

 For any tree partition $\T = (X,\leqslantt)$ and for all $x \in X$, $x$ must have at most one parent in $(X,\leqslantt)$.
 Informally, then, to construct a tree partition $\T$ from $\poset=(X,\leqslant)$, for all $x \in X$ we must discard all but (at most) one parent of $x$ in $\poset$.
 Hence, if we can associate the choice of parent $y$ for $z$ with an appropriate cost of the edge $yz$ in $H^*=(X,\eofhstar)$, then computing a minimal tree partition can be translated into a problem of selecting a suitable weighted forest.

 We now describe how to compute such a cost function.
 Given an information flow policy $\poset=(X,\leqslant)$, for each pair $yz$ such that $y > z$, we define 
  $\gamfuncp(yz) = \set{x \in X : x \geqslant z, x \not\geqslant y}$. 
  
 \begin{proposition}\label{pro:gamma-supset-inclusion-transitivity}
  For all $x > y > z$, $\gamfuncp(xz) \supset \gamfuncp(yz)$.
 \end{proposition}
 
 \begin{proof}
  Let $t \in \gamfuncp(yz)$.
  Then $t \geqslant z$ and $t \not\geqslant y$.
  Now if $t \geqslant x$, we would have $t \geqslant y$, by transitivity.
  Thus $t \not\geqslant x$ and hence $t \in \gamfuncp(xz)$.
  Moreover, $y \in \gamfuncp(xz)$, since $y > z$ and $y \not\geqslant x$, and $y \not\in \gamfuncp(yz)$, so the inclusion is strict.
 \end{proof}


  

Define a weight function $\omega_{\poset} : X \times X \rightarrow \mathbb{N}$, where 
\[
\omega_{\poset}(yz) =%
\begin{cases}
 \sum_{x \in \gamma_{\poset}(yz)} \card{U_x} & \text{if $y > z$}, \\
 0 & \text{otherwise}.
 \end{cases}
\]
 
 Note that for any tree partition $\T$, $z$ has at most one parent in $\T$, so we may write $\gamfunc{\T}(z)$ for $\gamfunc{\poset}(\parent{z}{\T}z)$ without ambiguity.
 Given a tree partition $\T$ of $X$, we define the weight function \mbox{$\Omega_{\T} : X \rightarrow \mathbb{N}$}, where
 \[
  \Omega_{\T}(z) =%
   \begin{cases}
    \displaystyle \sum_{x \geqslant z} \card{U_x} & \text{if $z$ is maximal in $\T$}, \\
    \displaystyle \sum_{x \in \gamma_{\T}(z)} \card{U_x} & \text{otherwise}.
   \end{cases}
 \]

Informally, $\Omega_{\T}(z)$ represents the number of users that will require the secret associated with $z$, on the one hand if $z$ is maximal in $\T$ and on the other if edge $\parentt{z}z$ is used in $\T$.
We can now prove the main result of this section, which establishes a relationship between \mbox{$\numsecrets{\bestkaf{\T}}{\T}$} and $\Omega_{\T}$, and thus enables us to define an (efficient) algorithm for computing a minimal tree partition.

 \newcommand{\mintree}{\widehat{\T}}
 
 \begin{theorem}\label{thm:building-optimal-tree-partition}
  Let $\poset = (X,\leqslant)$ be a poset with Hasse diagram $H(\poset) = (X,E_{\min})$ and let  $\T$ be a tree partition $\T$ of $\poset$.
  Then 
\begin{equation}\label{eq:totalsec}
   \numsecrets{\bestkaf{\T}}{\T} = \sum_{z \in X} \Omega_{\T}(z).
\end{equation}
Moreover, we can compute a minimal tree partition $\mintree$ of $\poset$ in time $O(\card{E_{\min}} + \card{X}^2)$.
 \end{theorem}
 
 \begin{proof}
  We first prove (\ref{eq:totalsec}).  Let $X''$ denote the set of maximal elements in $\T$ and $X'$ denote the set of non-maximal elements.
  By definition,
   \[
     \numsecrets{\bestkaf{\T}}{\T} = \sum_{x\in X} |\bestkaf{\T}(x)| |U_x|
   \]
   and, by Lemma~\ref{lem:phi-x-equals-gamma-y-z}, we have
   \[
    \card{\bestkaf{\T}(x)} = \card{\set{z \in X'\setminus \{x\}:\ x \in \gamfunc{\T}(z)}} + \card{\set{z \in X'' :\ x > z}} + 1.
   \]
   Hence
    \begin{align*}
	\numsecrets{\bestkaf{\T}}{\T} &= \sum_{x \in X} (\card{\set{z \in X': x \in \gamfunc{\T}(z)}} + \card{\set{z \in X'' : x > z}} + 1) |U_x| \\
					&= \sum_{x \in X} \card{\set{z \in X' : x \in \gamfunc{\T}(z)}} \card{U_x} - \sum_{x \in X'}|U_x| + \sum_{x \in X} \card{\set{z \in X'' : x > z}} \card{U_x} + \sum_{x \in X} \card{U_x} \\
					&= \sum_{x \in X} \card{\set{z \in X' : x \in \gamfunc{\T}(z)}} \card{U_x} + \sum_{x \in X} \card{\set{z \in X'' : x > z}} \card{U_x} + \sum_{x \in X''} \card{U_x} \\
	&=\sum_{z\in X'} \sum_{x\in \gamfunc{\T}(z)} \card{U_x} + \sum_{z\in X''} \sum_{x \geqslant z} \card{U_x} \\
	&=\sum_{z\in X} \Omega_{\T}(z) 
   \end{align*}

  We next establish the choice of $\T$ that minimizes $\numsecrets{\bestkaf{\T}}{\T}$. 
  Observe that if $z$ is not a maximal element of $X$, a minimal tree partition $\mintree$ will not have $z$ as a maximal element either. 
  Indeed, suppose $z$ is a maximal element in a tree partition $\T$ and let $y$ be a parent of $z$ in $X$. 
  Then $\Omega_{\T}(z)>\Omega_{{\T}'}(z)$, where ${{T}'}$ is obtained from $\T$ by adding edge $yz$ to the Hasse diagram of $\T$, since $\{x\in X:\ x\in \gamfunc{{\T}'}(z)\}\subset \{x\in X:\ x\geqslant z\}$; the inclusion is strict since $y$ is in the first set but not the second. 
  Thus, $z$ is a maximal in $\mintree$ if and only if $z$ is maximal in $X$. 
  It remains to decide on parents in $\mintree$ for non-maximal elements in $X$.

  Let $\T$ be a tree partition and $z$ is not maximal in $\T$. Note that $\Omega_{\T}(z)=\omega_{\poset}(\parent{z}{\T}z)$. 
  By Proposition~\ref{pro:gamma-supset-inclusion-transitivity}, we have $\gamfuncp(yz) \subset \gamfuncp(xz)$ for $x > y > z$.
  It follows that $\omega_{\poset}(yz) \leqslant \omega_{\poset}(xz)$, the inequality being strict if we assume that at least one user is assigned to each node in $X$.
  Thus it suffices to consider only parents of $z$ in $X$ when constructing a minimum tree partition.
  Moreover, to build $\mintree$, for each non-maximal $z \in X$, we select a parent $y$ of $z$ in $X$ such that $\omega_{\poset}(yz) \leqslant \omega_{\poset}(y'z)$ for all other parents $y'$ of $z$.
  
  Finally, we analyze the running time to compute a minimum tree partition.
  We can compute $\omega_{\poset}(yz)$ for each non-maximal $z$ and each parent $y$ of $z$ in $\poset$ in time $O(\card{X}^2)$ using an algorithm similar to that used for computing $\bestkaf{\T}$.
Now a minimal tree partition $\T$ of $\poset$ can be obtained by setting 
$y = \parentt{z}$, where $\weightfuncp(yz) \leqslant \weightfuncp(xz)$ for all $x \in X$ such that $x > z$.
This will require time $O(\card{E_{\min}})$. Thus, the total time required is $O(\card{E_{\min}}+|X|^2)$.\footnote{Since $\card{E_{\min}} \leqslant |X|^2$ we can simplify the total time to $O(\card{X}^2)$. However, we decided to keep $\card{E_{\min}}$ to stress that only parents of elements need to be considered to compute a minimum tree partition.}
 \end{proof}

We have shown that we can compute a minimal tree partition efficiently.
Recall that $\card{\bestkaf{\T}(x)}$ measures the number of secrets a user in $U_x$ will require to derive all authorized secrets (and keys).
We now consider whether it is possible to compute a minimal tree partition that simultaneously bounds $\max_{x \in X}\set{\card{\bestkaf{\T}(x)}}$.
Let $\T$ be a minimal tree partition of $\poset = (X, \leqslant)$.
We will say that $\T$ is an \emph{optimal tree partition} of $\poset$ if $\T$ has the minimum number of minimal elements among all minimal tree partitions.
An optimal tree partition with $\ell$ leaves has the property that no user will require more than $\ell$ secrets.

For each non-maximal $z \in \poset = (X,\leqslant)$, let $Y(z)$ be the set of $y\in X$ such that $y>z$ and $\weightfuncp(yz)$ is minimum. 
Construct a directed acyclic graph $H$ with vertex set $X$; for every non-maximal $y\in X$, the in-neighborhood of $y$ is $Y(y)$, and each maximal $y\in X$ has no in-neighbors. Add to $H$ a new vertex $r$ which is an in-neighbor of every $x\in X$.
Now apply the polynomial-time algorithm {\sc MinLeaf} \cite{GuRaKi09}, that allows us to find an out-tree rooted at $r$ with minimum number of leaves, i.e., vertices with no out-neighbors. As a result, we obtain, among all tree partitions with minimum number of secrets, one with minimum number of minimal elements. 
Let $X'$ denote the set of non-maximal elements in $\poset$. Then {\sc MinLeaf}'s runtime is $O(s+|X|^{3/2}s^{1/2})$, where $s=\sum_{z\in X'}|Y(z)|$. Observe that $s\leqslant \card{\eofhstar}$ and $\card{\eofhstar} \leqslant |X|^2$.   This implies that $O(s+|X|^{3/2}s^{1/2})=O(\card{X}^{3/2}\card{\eofhstar}^{1/2}).$
Thus, we have the following result.

 \begin{corollary}\label{cor:minimizing-leaves}
 Given an information flow policy $\poset = (X, \leqslant)$, we can find an optimal tree partition  $\T = (X, \leqslantt)$ of $\poset$ in time $O(\card{X}^{3/2}\card{\eofhstar}^{1/2})$.    
 \end{corollary}
 
 We conclude this section with an example illustrating our results.
  Let $[n] = \set{1,2,\dots,n}$ and let $[i,j] = \set{i,i+1,\dots,j-1,j}$ for $i \leq j$. 
 Then define the poset \[ \intn(n) = \set{[i,j] : 1 \leqslant i \leqslant j \leqslant n}, \] where $[i,j] \leqslant [i',j']$ if and only if $i' \leqslant i$ and $j' \geqslant j$.
 The Hasse diagram for $\intn(5)$ is illustrated in Figure~\ref{fig:hasse-diagram-poset}.
 The poset $\intn(n)$ has attracted considerable interest because of its application to ``time-bound'' access control (see~\cite{AtBlFr07,Cr11},  for example).
 In particular, the numbers $1,\dots,n$ represent time points or time intervals, and elements in $\intn(n)$ represent contiguous intervals of time (either consecutive points or a sequence of consecutive intervals).
 A user $u$ assigned the interval $[i,j]$ is authorized to access any object assigned an interval $[i',j'] \subseteq [i,j]$. 
 
 The cardinality of $\gamma_{\poset}(yz)$, $y,z \in \intn(5)$, $y \gtrdot z$, is shown in Figure~\ref{fig:intervals-5-edge-weights}.
 A tree of minimum weight is shown in Figure~\ref{fig:interval-5-tree-minimum-weight} and the corresponding values of $\Omega_{\T}(z)$ are shown in Figure~\ref{fig:interval-5-gamma-z}.
 It is possible to show that the minimum number of secrets required in total, assuming $\card{U_x} = 1$ for each $x \in \intn(n)$, is $\frac{1}{6}m(m+1)(4m-1)$ if $n = 2m-1$, and $\frac{1}{6}m(m+1)(4m+5)$ if $n = 2m$.

 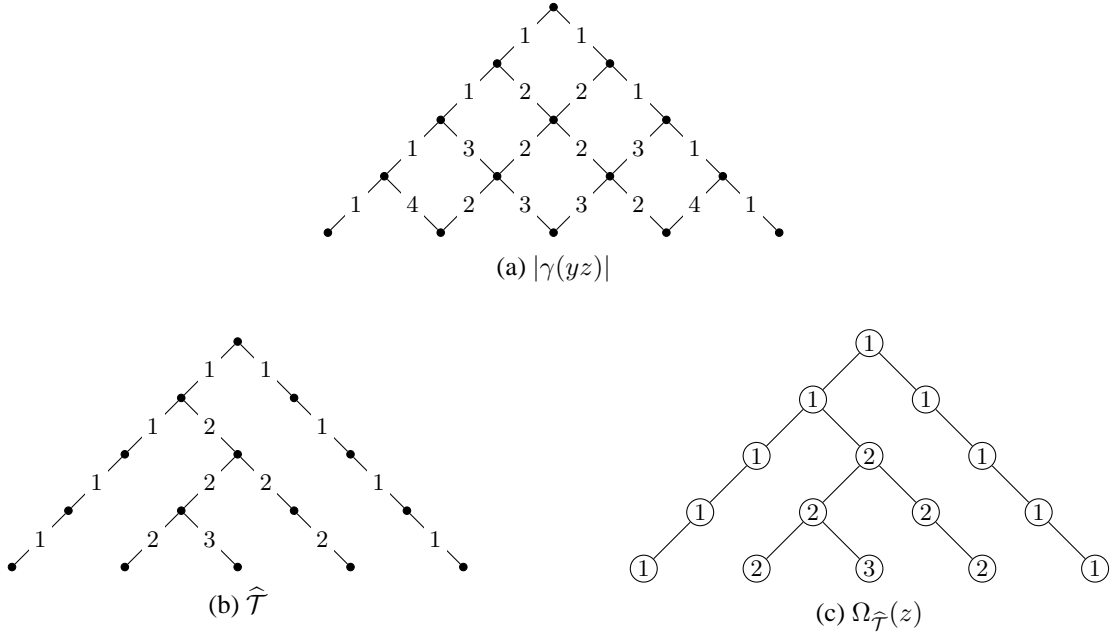
\begin{figure}[h]\centering
  \begin{subfigure}{.475\textwidth}\centering
   \newcommand{\addwt}[3]{\draw (#1) to node[wt] {$#2$} (#3)}
   \begin{tikzpicture}[vtx/.style={circle,draw,fill,inner sep=1pt},wt/.style={fill=white,font=\footnotesize,inner sep=2pt},x=1cm,y=.5cm,scale=.75]
    \node[vtx] (11) at (0,0) {};
    \node[vtx] (22) at (2,0) {};
    \node[vtx] (33) at (4,0) {};
    \node[vtx] (44) at (6,0) {};
    \node[vtx] (55) at (8,0) {};
    \node[vtx] (12) at (1,2) {};
    \node[vtx] (23) at (3,2) {};
    \node[vtx] (34) at (5,2) {};
    \node[vtx] (45) at (7,2) {};
    \node[vtx] (123) at (2,4) {};
    \node[vtx] (234) at (4,4) {};
    \node[vtx] (345) at (6,4) {};
    \node[vtx] (1234) at (3,6) {};
    \node[vtx] (2345) at (5,6) {};
    \node[vtx] (12345) at (4,8) {};
    \draw (12345) to node[wt] {$1$}  (1234);
    \draw (12345) to node[wt] {$1$}  (2345);
    \draw (1234) to node[wt] {$1$}  (123);
    \draw (1234) to node[wt] {$2$}  (234);
    \draw (2345) to node[wt] {$2$}  (234);
    \draw (2345) to node[wt] {$1$}  (345);
    \draw (123)  to node[wt] {$1$}  (12);
    \draw (123)  to node[wt] {$3$} (23);
    \draw (234)  to node[wt] {$2$} (23);
    \draw (234)  to node[wt] {$2$} (34);
    \draw (345)  to node[wt] {$3$} (34);
    \draw (345)  to node[wt] {$1$} (45);
    \draw (12)  to node[wt] {$1$} (11);
    \draw (12)  to node[wt] {$4$} (22);
    \draw (23)  to node[wt] {$2$} (22);
    \draw (23)  to node[wt] {$3$} (33);
    \draw (34)  to node[wt] {$3$} (33);
    \draw (34)  to node[wt] {$2$} (44);
    \draw (45)  to node[wt] {$4$} (44);
    \draw (45)  to node[wt] {$1$} (55);
    \end{tikzpicture}
    \caption{$\card{\gamma(yz)}$}\label{fig:intervals-5-edge-weights}
   \end{subfigure}
  
   \vspace*{\baselineskip}
   
  \begin{subfigure}{.475\textwidth}\centering
   \newcommand{\addwt}[3]{\draw (#1) to node[wt] {$#2$} (#3)}
   \begin{tikzpicture}[vtx/.style={circle,draw,fill,inner sep=1pt},wt/.style={fill=white,font=\footnotesize,inner sep=2pt},x=1cm,y=.5cm,scale=.75]
    \node[vtx] (11) at (0,0) {};
    \node[vtx] (22) at (2,0) {};
    \node[vtx] (33) at (4,0) {};
    \node[vtx] (44) at (6,0) {};
    \node[vtx] (55) at (8,0) {};
    \node[vtx] (12) at (1,2) {};
    \node[vtx] (23) at (3,2) {};
    \node[vtx] (34) at (5,2) {};
    \node[vtx] (45) at (7,2) {};
    \node[vtx] (123) at (2,4) {};
    \node[vtx] (234) at (4,4) {};
    \node[vtx] (345) at (6,4) {};
    \node[vtx] (1234) at (3,6) {};
    \node[vtx] (2345) at (5,6) {};
    \node[vtx] (12345) at (4,8) {};
    \draw (12345) to node[wt] {$1$}  (1234);
    \draw (12345) to node[wt] {$1$}  (2345);
    \draw (1234) to node[wt] {$1$}  (123);
    \draw (1234) to node[wt] {$2$}  (234);
    \draw (2345) to node[wt] {$1$}  (345);
    \draw (123)  to node[wt] {$1$}  (12);
    \draw (234)  to node[wt] {$2$} (23);
    \draw (234)  to node[wt] {$2$} (34);
    \draw (345)  to node[wt] {$1$} (45);
    \draw (12)  to node[wt] {$1$} (11);
    \draw (23)  to node[wt] {$2$} (22);
    \draw (23)  to node[wt] {$3$} (33);
    \draw (34)  to node[wt] {$2$} (44);
    \draw (45)  to node[wt] {$1$} (55);
    \end{tikzpicture}
    \caption{$\mintree$}\label{fig:interval-5-tree-minimum-weight}
   \end{subfigure}
  \hfill
  \begin{subfigure}{.475\textwidth}\centering
   \newcommand{\addwt}[3]{\draw (#1) to node[wt] {$#2$} (#3)}
   \begin{tikzpicture}[vtx/.style={circle,draw,inner sep=1pt,font=\footnotesize},wt/.style={fill=white,font=\footnotesize,inner sep=2pt},x=1cm,y=.5cm,scale=.75]
    \node[vtx] (11) at (0,0) {$1$};
    \node[vtx] (22) at (2,0) {$2$};
    \node[vtx] (33) at (4,0) {$3$};
    \node[vtx] (44) at (6,0) {$2$};
    \node[vtx] (55) at (8,0) {$1$};
    \node[vtx] (12) at (1,2) {$1$};
    \node[vtx] (23) at (3,2) {$2$};
    \node[vtx] (34) at (5,2) {$2$};
    \node[vtx] (45) at (7,2) {$1$};
    \node[vtx] (123) at (2,4) {$1$};
    \node[vtx] (234) at (4,4) {$2$};
    \node[vtx] (345) at (6,4) {$1$};
    \node[vtx] (1234) at (3,6) {$1$};
    \node[vtx] (2345) at (5,6) {$1$};
    \node[vtx] (12345) at (4,8) {$1$};
    \draw (12345) -- (1234);
    \draw (12345) -- (2345);
    \draw (1234) -- (123);
    \draw (1234) -- (234);
    \draw (2345) -- (345);
    \draw (123) -- (12);
    \draw (234) -- (23);
    \draw (234) -- (34);
    \draw (345) -- (45);
    \draw (12) -- (11);
    \draw (23) -- (22);
    \draw (23) -- (33);
    \draw (34) -- (44);
    \draw (45) -- (55);
    \end{tikzpicture}
    \caption{$\Omega_{\mintree}(z)$}\label{fig:interval-5-gamma-z}
   \end{subfigure}
  \caption{A minimal tree partition of $(\intn(5),\subseteq)$}\label{fig:tree-partition-interval-5}
 \end{figure}

\section{Selecting a Good Chain Partition}\label{sec:chain-based-schemes}

\newcommand{\treepi}{\widetilde{\Pi}}
\newcommand{\eofpi}{E_{0,\Pi}}
\newcommand{\eoftreepi}{E_{0,\treepi}}

\newcommand{\eofg}{E_0^*}
\newcommand{\eofc}{E_C}

\newcommand{\vin}[1][x]{#1_{\rm in}}
\newcommand{\vout}[1][x]{#1_{\rm out}}

In this section, we consider chain-based schemes. Recall that a chain partition of a poset $\poset$ is a disjoint union of chains such that every element
of $\poset$ belongs to one of the chains. An element $z$ of a chain $C$ is called \emph{top} (\emph{bottom}, respectively) if the in-degree (out-degree, respectively) of $z$ in $H(C)$ is zero.

 


We first show that the number of secrets  to be issued in a chain-based enforcement scheme is determined by the bottom elements of the chains in the corresponding chain partition.
This in turn implies that there exists a chain partition with a minimum number of secrets issued for which the number of chains is exactly the width of the poset.


  

\begin{lemma}\label{lem:number-of-keys-from-bottom-elements}
For any poset $\poset=(X, \leqslant)$ and any chain partition $\C = (X, \leqslantc)$ of $(X, \leqslant)$ with chains $\set{C_1,\dots,C_\ell}$, let chain $C_i$ have bottom element $b_i$, $1 \leqslant i \leqslant \ell$.
 Then 
 \begin{equation}\label{eq:chainsec}
  \numsecrets{\bestkaf{\C}}{\C}  = \sum_{i=1}^\ell \sum_{x \in \uset{b_i}{\poset}}\card{U_x}.
\end{equation}
\end{lemma}

\begin{proof}
Let $C_i$ comprise elements $z_1,z_2, \dots ,z_c$ such that $z_1>z_2>\dots >z_c$ (i.e., $b_i=z_c$) and observe that 
$\uset{b_i}{\poset}$ is the disjoint union of sets $X_i$, $1 \leqslant i \leqslant c$, where $X_1=\{x:\ x\geqslant  z_1\}$ and $X_j=\{x:\ x \not\geqslant  z_{j-1}, x \geqslant  z_j\},$  $2 \leqslant j \leqslant c$. Observe that $X_j=\{x:\ x\in \gamma_{\C}(z)\}$, $2 \leqslant j \leqslant c$. This decomposition of $\uset{b_i}{\poset}$ into sets $X_i$, $1 \leqslant i \leqslant c$, will be used in the following derivation. 

By (\ref{eq:totalsec}) and the definition of $\Omega_{\C}(z)$, 
\begin{align*} 
  \numsecrets{\bestkaf{\C}}{\C} &= \sum_{z\in X} \Omega_{\C}(z)\\
                                                  &= \sum_{i = 1}^{\ell}\sum_{x\in X_1}|U_x| + \sum_{i = 1}^{\ell} \sum_{j=2}^{\ell} \sum_{x\in X_j}|U_x| \\
                                                  &= \sum_{i = 1}^{\ell} \sum_{x \in \uset{b_i}{\poset}}\card{U_x}\qedhere
\end{align*}
\end{proof}


By Dilworth's Theorem, a poset $(X,\leqslant)$ of width $w$ of has a chain partition with $w$ chains. Such a chain partition can be obtained in time $O(|X|^{2.5})$ \cite{Garg15}. Thus, in particular, we can compute $w$ in time $O(|X|^{2.5})$.
The next theorem can be viewed as a strengthening of Dilworth's Theorem. In Subsection \ref{sec:chain-partition-requiring-widehat-k-keys}, we will show how to compute a minimal chain partition of width $w$
in polynomial time.
  
\begin{theorem}\label{thm:chain-partition-only-requires-w-chains}
 Let $\poset=(X,\leqslant)$ be an information flow policy of width $w$.
 Then there exists a minimal chain partition of width $w$.
\end{theorem}

\begin{proof}
 Let $\C=(X,\leqslantc)$ be a minimal chain partition of $X$ into $t \geqslant w$ chains and let $B$ be the set of bottom elements in the chains of $\C$.
 A theorem of Gallai and Milgram asserts that if a chain partition $\C$ of a poset $\poset$ contains $t$ chains, where $t > w$, then there exists a chain partition $\C'=(X,\leqslant_{C'})$ into $t-1$ chains such that the set of bottom elements in  $\C'$  is a subset of $B$~\cite{GaMi60}.\footnote{The result is phrased in the language of digraphs, but every poset may be represented by an equivalent transitive acyclic digraph.}
 Hence, by iterated applications of the Gallai-Milgram theorem, there exists a chain partition $\C^*=(X,\leqslant_{C^*})$ of width $w$ such that the set of bottom elements $B^*$ in $\C^*$ is a subset of $B$.
 Moreover, by Lemma~\ref{lem:number-of-keys-from-bottom-elements},
 \[
  \numsecrets{\bestkaf{\C^*}}{\C^*} = \sum_{b \in B^*}\sum_{x \in \uset{b}{\poset}}\card{U_x}  \leqslant  \sum_{b \in B} \sum_{x \in \uset{b}{\poset}}\card{U_x} =  \numsecrets{\bestkaf{\C}}{\C}
 \]
 As $\C$ is a minimal chain partition, we conclude that $\C^*$ is also a minimal chain partition.
\end{proof}

\begin{corollary}\label{cor:users-only-require-w-secrets}
Let $\poset=(X,\leqslant)$ be an information flow policy.
There exists a chain partition $\C=(X,\leqslantc)$ such that $\numsecrets{\bestkaf{\C}}{\C}$ is minimized  
and $\max\set{\card{\bestkaf{\C}(x)} : x \in X} \leqslant w$.
\end{corollary}

\begin{proof}
 The result follows immediately from Theorem~\ref{thm:chain-partition-only-requires-w-chains}
and the fact that $\card{\bestkaf{\C}(x)}$ is bounded above by the number of chains in  $\C$  for all $x \in X$.
\end{proof}

The above corollary shows that no user requires more than $w$ secrets in a chain-based enforcement scheme.

Returning to our example of $\intn(n)$, note that the width of $\intn(n)$ is $n$ as the minimal elements form the largest antichain.
Thus, any chain partition with $n$ chains requires the same number of secrets. It is not hard to show that this number is $\frac{1}{6}n(n+1)(n+2)$, which is minimum possible.
 Thus the minimal tree partition of $\intn(n)$ (discussed in Section~\ref{sec:tree-based-schemes}) requires approximately half the number of secrets required by the minimal chain partition.

\subsection{Computing a minimal chain partition}\label{sec:chain-partition-requiring-widehat-k-keys}

A chain partition imposes stronger constraints than a tree partition.
Specifically, each element in a chain partition has at most one parent and one child, whereas a tree partition only requires that each element has at most one parent.
Thus, the straightforward algorithm for computing a minimal tree partition cannot be used to compute a minimal chain partition.

Suppose $\poset=(X,\leqslant)$ is a poset of width $w$.
In general, a chain partition of $\poset$ has $\ell \geqslant w$ chains.
Theorem~\ref{thm:chain-partition-only-requires-w-chains} asserts that there exists a 
minimal chain partition comprising $w$ chains.
We now show how such a chain partition may be constructed.
In particular, we show how to transform the problem of finding a minimal chain partition $\C=(X,\leqslantc)$
into a problem of finding a minimum cost flow in a network.

Informally, a \emph{network} is a directed graph in which each edge is associated with a \emph{capacity}.
A \emph{network flow} associates each edge in a given network with a flow, which must not exceed the capacity of the edge.
Networks are widely used to model systems in which some quantity passes through channels (edges in the network) that meet at junctions (vertices); examples include traffic in a road system, fluids in pipes, or electrical current in circuits.
Our definitions for networks and network flows follow the presentation of Bang-Jensen and Gutin~\cite{BaGu02}.

\begin{definition}
 A \emph{network} is a tuple $\mathcal{N} = (D,l,u,c,\beta)$, where:
 \begin{itemize}
  \item $D = (V,A)$ is a directed graph with vertex set $V$ and edge set $A$;
  \item $l : V \times V \rightarrow \mathbb{N}$ such that $l(vv') = 0$ if $vv' \not\in A$ and $l(vv') \geqslant 0$ otherwise;
  \item $u : V \times V \rightarrow \mathbb{N}$ such that $u(vv') = 0$ if $vv' \not\in A$ and $u(vv') \geqslant l(vv') \geqslant 0$ otherwise;
  \item $c : V \times V \rightarrow \mathbb{R}$;
  \item $\beta : V \rightarrow \mathbb{R}$ such that $\sum_{v \in V} \beta(v) = 0$.
 \end{itemize}
\end{definition}

Intuitively, $l$ and $u$ represent lower and upper bounds, respectively, on how much flow can pass through each edge, and $c$ represents the cost associated with each unit of flow in each edge. The function $\beta$ represents how much flow should enter or leave the network at a given vertex. If $\beta(x)=0$, then the flow going into $x$ should be equal to the flow going out of $x$. If $\beta(x)>0$, then there should be $\beta(x)$ more flow coming out of $x$ than going into $x$. If $\beta(x)<0$, there should be $|\beta(x)|$ more flow going into $x$ than coming out of $x$.

\begin{definition}
Given a network $\mathcal{N} = (D, l,u, c, \beta)$, a function $f:V  \times V \rightarrow \mathbb{N}$ is a \emph{feasible flow} for $\mathcal{N}$ if the following conditions are satisfied:
 \begin{itemize}
  \item $u(vv') \geqslant f(vv') \geqslant l(vv')$ for every $vv' \in V \times V$;
  \item $\sum_{v' \in V}(f(vv') - f(v'v)) = \beta(v)$ for every $v \in V$.
 \end{itemize}
The \emph{cost} of $f$ is defined to be \[ \sum_{vv' \in A} c(vv') f(vv'). \]
\end{definition}

Our aim is to find a tree $\C = (X,\leqslantc)$ such that $\C$ is a chain partition of $X$ with precisely $w$ chains that minimizes $\numsecrets{\bestkaf{\C}}{\C}$.
To do this, we will construct a network $\mathcal{N}$ such that the minimum cost flow of $\mathcal{N}$ corresponds to the desired chain partition.
We can then find the minimum cost flow of $\mathcal{N}$ in polynomial time.

Every top vertex in $\C$ must have one child and no parent in $C$, every bottom vertex in $C$ must have one parent and no child in $C$, and every other vertex in $C$ must have one parent and one child.
We cannot represent this requirement directly in a network.
However, we can use the \emph{vertex splitting procedure}~\cite{BaGu02} to simulate it.
Specifically, given poset $\poset=(X,\leqslant)$, define first a directed graph $D = (V,A)$.
Let $\vin[X] = \set{\vin : x \in X}$ and $\vout[X] = \set{\vout : x \in X}$, and define the vertex set $V = \vin[X] \cup \vout[X] \cup \set{s,t}$, where $\{s,t\}\cap (\vin[X] \cup \vout[X])=\emptyset.$
Define 
the edge set $A$ as follows: for $v, v' \in \vin[X] \cup \vout[X]$, $vv' \in A$ if and only if either $v = \vin$ and $v'=\vout$ for some $x \in X$, 
or $v = \vout[x]$ and $v' = \vin[y]$ for some $x,y \in X$ such that $y \leqslant x$;
for every $v \in \vin[X]$ we have $s v \in A$; and
for every $v \in \vout[X]$ we have $v t \in A$.

Then define a network $\mathcal{N}=(D,l,u,c,\beta)$, where
\begin{align*}
 l(vv') &= %
   \begin{cases} 
    1 & \text{if $v=\vin, v'=\vout,$ where $x \in X$} \\
    0 & \text{otherwise;}
   \end{cases} \\
 u(vv') &= 
   \begin{cases}
    1 & \text{if $vv' \in A$} \\
    0 & \text{otherwise;}
   \end{cases} \\
 c(vv') &= %
  \begin{cases}
    \sum_{x \in \uset{v}{\poset}}\card{U_x}& \text{if  $v' = t$, $v = \vout[x],$ where $x\in X$} \\
    0  & \text{otherwise;} \\
  \end{cases} \\
 \beta(v) &= 
 \begin{cases}
  w & \text{if $v = s$} \\
  -w & \text{if $v = t$} \\
  0 & \text{otherwise.}
 \end{cases}
\end{align*}
We call this network the \emph{network chain-representation of $(X,\leqslant)$}.
Note that any feasible flow $f$ for this network must have $0 \leqslant f(xy) \leqslant 1$ for all $xy \in A$. 

\begin{lemma}\label{lem:min-cost-flow-equals-min-keys}
 Let $\mathcal{N}$ be the network chain-representation of an information flow policy $\poset=(X,\leqslant)$.
 Then the minimum number of secrets issued by a chain-based enforcement scheme for $(X,\leqslant)$ with $w$ chains is $\widehat{f}$, where $\widehat{f}$ is the minimum cost of a feasible flow in $\mathcal{N}$.
\end{lemma}

\begin{proof}
 Suppose we are given a chain partition $\C=(X, \leqslantc)$.
 Consider the following flow:
 \begin{align*}
  f(\vin\vout) &= 1\qquad \text{for all $x \in X$}; \\
  f(\vout\vin[y]) &= 1\qquad \text{if $x = {\rm par}_{\C}(y)$}; \\
  f(s\vin) &= 1\qquad \text{if $x$ is the top element in a chain in  $\C$}; \\
  f(\vout t) &= 1\qquad \text{if $x$ is the bottom element in a chain in  $\C$}; \\
  f &= 0 \qquad \text{otherwise}.
 \end{align*}
 Observe that $f$ is a feasible flow.
Indeed, by construction all edges $xy$ satisfy  $u(xy) \geqslant f(xy) \geqslant l(xy)$. In the graph formed by edges $xy$ with $f(xy)=1$, it is clear that every vertex $x$ has in-degree and out-degree $1$, except for $s$ and $t$. Also, $s$ has in-degree $0$ and out-degree $w$ in this graph, and $t$ has in-degree $w$ and out-degree $0$.
As all edges $xy$ have $f(xy)=1$ or $f(xy)=0$, we have that \[ \sum_{v \in V(D)}(f(xv) - f(vx)) = \beta(x) \] for all $x$, as required.
Moreover, the cost of $f$ equals $\sum_{b\in B}\sum_{x \in \uset{b}{\poset}}\card{U_x} $, where $B$ is the set of bottom elements of chains in $\C$, which by (\ref{eq:chainsec}) equals $\mathcal{S}(\C,\phi_{\C})$.

Conversely, suppose $f$ is a feasible flow for $\mathcal{N}$.
Then we define $y\lessdotc x$ if and only if $x,y \in X$ and $f(\vout \vin[y])=1$. By the construction of $\mathcal{N}$ and definition of $f$, it is not hard to see that $\C$ is a chain partition of $X$ with $w$ chains. By construction of $\mathcal{N}$, the cost of $f$ equals $\sum_{b\in B}\sum_{x \in \uset{b}{\poset}}\card{U_x} $, where $B$ is the set of bottom elements of chains in $\C$, which by (\ref{eq:chainsec}) equals $\mathcal{S}(\C,\phi_{\C})$.
\end{proof}

\begin{lemma}\label{lem:min-cost-flow-in-poly-time}
 We can find a minimum cost flow for $\mathcal{N}$ in $O(|X|^4w)$ time.
\end{lemma}
\begin{proof} Recall that computing $w$ can be done in time $O(|X|^{2.5}).$ To compute $\sum_{x \in \uset{y}{\poset}}\card{U_x} $ for each $y\in X$ requires time $O(|E_{\max}|+|X|)$ using depth-first search from $y$ in the digraph obtained from $H^*(X)$ by changing orientation of every edge. Thus, to compute $\sum_{x \in \uset{y}{\poset}}\card{U_x} $ for all $y\in X$ requires time $O(|X|(|E_{\max}|+|X|)).$

The well-known buildup algorithm (see~\cite[\S4.10.5]{BaGu02}, for example) finds a minimum cost flow for a network with $n$ vertices and $m$ edges in time $O(n^2 m M)$, where $M$ denotes the maximum of all absolute values of balance demands on vertices.
By construction of $\mathcal{N}$, we have that $n = 2|X| + 2 = O(|X|)$, $m = O(n^2)=O(|X|^2)$, and $M = w$.
Thus we get the desired running time.
\end{proof}

\begin{remark}
Strictly speaking, the buildup algorithm assumes that all lower bounds on edges are $0$. 
In its current form, our network does not satisfy this condition.
However, we can satisfy this condition, given $\mathcal{N} = (D,l,u,c,\beta)$, by defining the network $\mathcal{N}' = (D,l',u',c,\beta')$, where
\begin{alignat*}{2}
 l'(xy) &= 0 & \qquad & \beta'(x) = \beta(x) - l(xy) \\
 u'(xy) &= u(xy) - l(xy) & \qquad & \beta'(y) = \beta(y) + l(xy) 
\end{alignat*}
Then the minimum cost flow $f'$ for $\mathcal{N}'$ will have cost exactly $\sum_{xy} l(xy)c(xy)$ less than the minimum cost flow for $\mathcal{N}$, and $f'$ can be transformed into a minimum cost feasible flow $f$ for $\mathcal{N}$ by setting $f(xy) = f'(xy)+l(xy)$.
\end{remark}

We are now able to prove our main result, for this section which is, essentially, a corollary of Theorem~\ref{thm:chain-partition-only-requires-w-chains} and Lemmas~\ref{lem:min-cost-flow-equals-min-keys} and~\ref{lem:min-cost-flow-in-poly-time}.

\begin{theorem}\label{thm:main-theorem}
 Let $\poset=(X,\leqslant)$ be an information flow policy of width $w$.
Then we can find a minimal chain partition comprising $w$ chains in time  $O(|X|^4w)$. In such a chain partition no user requires more than $w$ secrets.
\end{theorem}
\begin{proof}
Let ${\cal S}$ denote the minimum number of secrets issued by a chain-based enforcement scheme for $X$.
By Theorem \ref{thm:chain-partition-only-requires-w-chains}, there exists a chain partition  that has exactly $w$ chains, for which the corresponding chain-based enforcement scheme only requires ${\cal S}$ secrets.
Then by Lemma \ref{lem:min-cost-flow-equals-min-keys},  ${\cal S}$  is equal to the minimum cost of a feasible flow in $\mathcal{N}$, the network chain-representation of $\poset$.
By Lemma \ref{lem:min-cost-flow-in-poly-time}, such a flow can be found in $O(|X|^4w)$ time, and this flow can be easily transformed into the corresponding chain partition $\C=(X, \leqslantc)$. 
Finally, by definition of $\bestkaf{\C}(x)$, $|\bestkaf{\C}(x)| \leqslant w$ for each $x \in X$ and therefore no user requires more than $w$ secrets.
\end{proof}

%

\section{Concluding Remarks}\label{sec:conclusion}

In this paper, we introduced the concept of a tree partition, generalizing prior work on chain partitions and tree-based enforcement schemes.
We have proved that it is possible to compute optimal chain and tree partitions for an arbitrary information flow policy in polynomial time.
And we have proved that there exist secure instantiations of enforcement schemes based on tree partitions.
In short, we have shown that it is possible to construct forest-based cryptographic enforcement schemes for information flow policies efficiently.

Perhaps the most important contribution of our work on cryptographic enforcement schemes based on tree and chain partitions is to provide alternative trade-offs between the parameters of such enforcement schemes.
These additional trade-offs provide data owners with a greater range of potential enforcement schemes, enabling them to select the most appropriate for their particular information flow policy and deployment constraints (such as storage and connectivity capabilities of end-user devices).
We might, for example, wish to use an existing scheme that requires each device to store a single secret when storage is limited.
Alternatively, we might wish to use a chain-based scheme when the distribution of public information is difficult and we wish to impose a small upper bound on the number of secrets that any device needs to store.
We might use a tree-based scheme if distribution of public information is difficult and we wish to minimize the amount of data we wish to transmit to the user population.

%
%

Another difference between minimal tree-based and chain-based schemes is that computing the former is significantly faster than the latter as the former can essentially be computed by a simple greedy  algorithm, while the latter requires a more sophisticated and much slower minimum cost flow algorithm. While still polynomial-time, minimum cost flow algorithms may be too slow when $|X|$ is large.

In future work, we hope to investigate the difficulty of finding a tree partition in which the worst-case derivation time is as similar as possible for all users (whilst still minimizing the number of secrets issued).

\bibliography{refs}
\bibliographystyle{abbrv}

\end{document}